\documentclass[letterpaper, 10 pt, conference]{ieeeconf}%

\IEEEoverridecommandlockouts
\usepackage{cite}
\usepackage{amsmath,amssymb,amsfonts}
\usepackage{algorithmic}
\usepackage{graphicx}
\usepackage{textcomp}
\usepackage{xcolor}
\usepackage{cuted}
\usepackage{lipsum}
\usepackage{mathtools}
\usepackage{stfloats}
\usepackage{bbm}
\usepackage{multirow}

\newcommand{\R}{\mathbb{R}}
\newcommand{\N}{\mathbb{N}}
\newcommand{\T}{\top}
\newcommand{\I}{\mathbf{I}}
\newcommand{\0}{\mathbf{0}}
\newcommand{\E}{\mathcal{E}}

\newcommand{\diag}{\text{diag}}
\newcommand{\tsup}[1]{\textsuperscript{#1}}
\newcommand{\mb}[1]{\mathbf{#1}}

\newcommand{\bm}[1]{\begin{bmatrix}#1\end{bmatrix}}
\def\BibTeX{{\rm B\kern-.05em{\sc i\kern-.025em b}\kern-.08em
    T\kern-.1667em\lower.7ex\hbox{E}\kern-.125emX}}

\newtheorem{assumption}{\textbf{Assumption}}
\newtheorem{theorem}{\textbf{Theorem}}
\newtheorem{corollary}{\textbf{Corollary}}
\newtheorem{lemma}{\textbf{Lemma}}
\newtheorem{proposition}{\textbf{Proposition}}
\newtheorem{remark}{\textbf{Remark}}
\newtheorem{definition}{\textbf{Definition}}

\usepackage{hyperref}

\begin{document}

\title{
{\LARGE \textbf{
Dissipativity-Based Distributed Control and Communication Topology Co-Design for Voltage Regulation and Current Sharing in DC Microgrids
}}}



\author{Mohammad Javad Najafirad and Shirantha Welikala 
\thanks{The authors are with the Department of Electrical and Computer Engineering, School of Engineering and Science, Stevens Institute of Technology, Hoboken, NJ 07030, \texttt{{\small \{mnajafir,swelikal\}@stevens.edu}}.}}

\maketitle


\begin{abstract}
This paper presents a novel dissipativity-based distributed droop-free control approach for voltage regulation and current sharing in DC microgrids (MGs) comprised of an interconnected set of distributed generators (DGs), loads, and power lines. First, we describe the closed-loop DC MG as a networked system where the DGs and lines (i.e., subsystems) are interconnected via a static interconnection matrix. This interconnection matrix demonstrates how the inputs, outputs, and disturbances of DGs and lines are connected in a DC MG. Each DG is equipped with a local controller for voltage regulation and a distributed global controller for current sharing, where the local controllers ensure individual voltage tracking while the global controllers coordinate among DGs to achieve proportional current sharing. To design the distributed global controllers, we use the dissipativity properties of the subsystems and formulate a linear matrix inequality (LMI) problem. To support the feasibility of this problem, we identify a set of necessary local and global conditions to enforce in a specifically developed LMI-based local controller design process. 
In contrast to existing DC MG control solutions, our approach proposes a unified framework for co-designing the distributed controller and communication topology. 
As the co-design process is LMI-based, it can be efficiently implemented and evaluated using existing convex optimization tools. The effectiveness of the proposed solution is verified by simulating an islanded DC MG in a MATLAB/Simulink environment under different scenarios, such as load changes and topological constraint changes, and then comparing the performance with the droop control algorithm. 
\end{abstract}


\section{Introduction} 
The microgrid (MG) concept has been introduced as a comprehensive framework for the cohesive coordination of distributed generators (DGs), variable loads, and energy storage units within a controllable electrical network to facilitate the efficient integration of renewable energy resources such as wind turbines and photovoltaic systems \cite{liu2023resilient}. DC MGs have gained more attention in recent years due to the growing demand for DC loads such as data centers, electric vehicle chargers, and LED lighting. In addition, DC MGs offer distinct advantages over AC systems by eliminating unnecessary conversion stages and removing frequency regulation \cite{dou2022distributed}.

The two primary control goals in DC MGs are voltage regulation and current sharing. To achieve these goals, centralized \cite{mehdi2020robust}, decentralized \cite{peyghami2019decentralized}, and distributed control \cite{xing2019distributed} are proposed. Although the centralized approach provides controllability and observability, it suffers from a single point of failure \cite{guerrero2010hierarchical}. In decentralized control, only a local controller is required; hence, there is no communication among DGs \cite{khorsandi2014decentralized} that compromises the proportional current sharing. The lack of coordination can be solved by developing distributed control in which the DGs can share their variables with their neighbors through a communication network \cite{dehkordi2016distributed}.

The conventional decentralized control approach is droop control, where each DG operates independently based on local measurements without communication among units. However, due to line impedance mismatch and droop characteristics, pure decentralized droop control cannot simultaneously achieve voltage regulation and current sharing. To address this limitation, distributed control solutions have been developed, which can be categorized into two main approaches. The first approach maintains the droop mechanism but adds distributed secondary control layers. In hierarchical droop-based systems \cite{zhou2020distributed}, a distributed secondary control is applied to achieve average voltage regulation and proper load sharing \cite{nasirian2014distributed}. While these distributed approaches improve performance over pure decentralized methods, they still fundamentally rely on droop mechanisms and require careful tuning of droop coefficients to balance conflicting objectives. The second approach completely eliminates the droop mechanism in favor of droop-free distributed control algorithms. A droop-free optimal feedback control is introduced in \cite{dissanayake2019droop} to replace the droop mechanism entirely. In \cite{zhang2022droop}, a dynamic consensus algorithm is proposed where droop-free cooperative control is constructed. These approaches rely solely on distributed coordination through communication networks to achieve both voltage regulation and current sharing objectives.

\begin{figure}
    \centering
    \includegraphics[width=0.99\columnwidth]{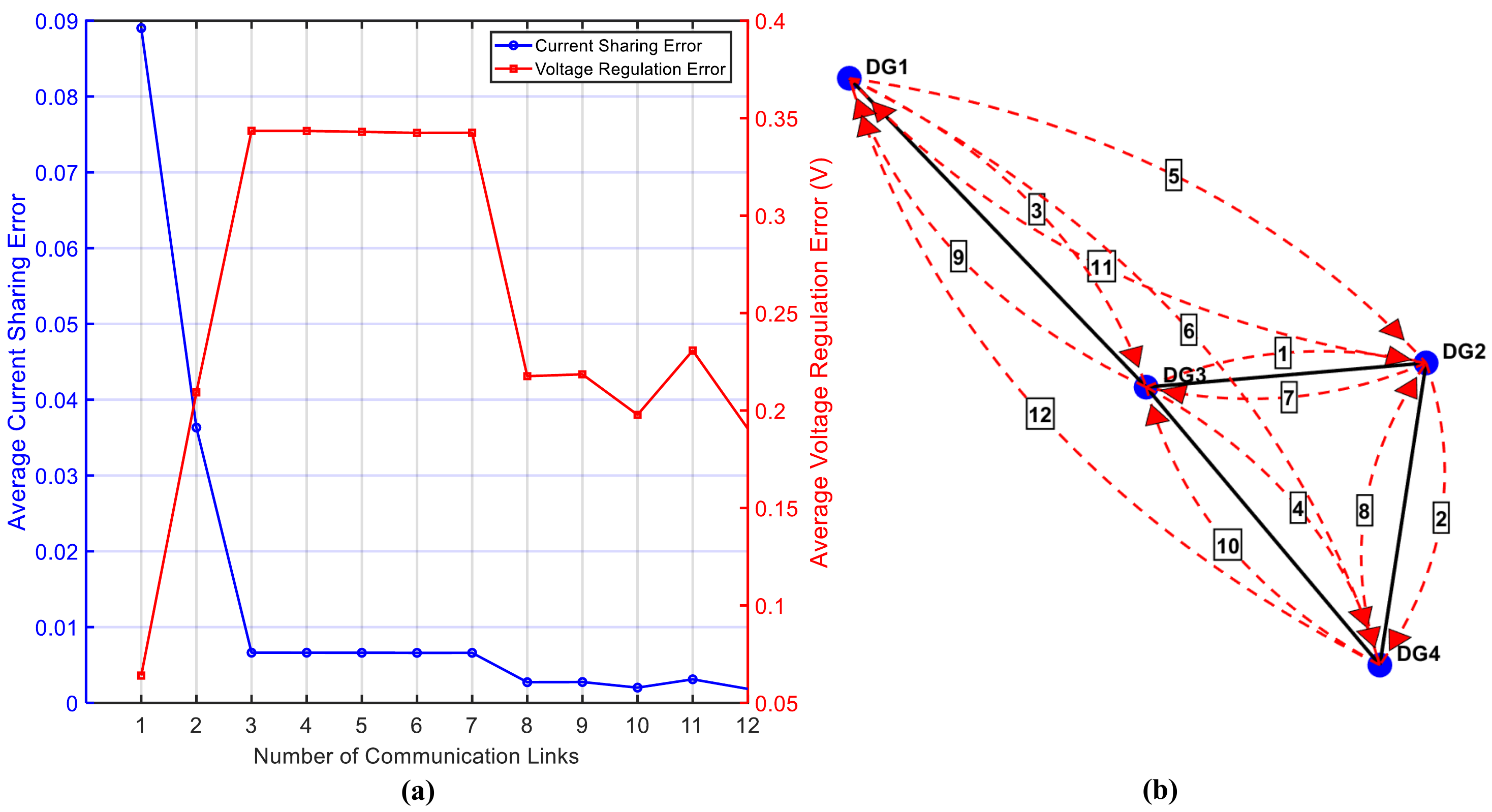}
    \caption{The impact of communication links on DC MG performance: (a) average voltage and current sharing error for different numbers of communication links and (b) corresponding communication network topology.}
    \label{fig.initialresult}
\end{figure}

However, conventional distributed controller design proceeds independently from communication topology considerations, with communication network structures often assumed to be static and predetermined. Many existing approaches simply align the communication topology with the physical electrical connections (i.e., DGs communicate only with electrically connected neighbors), without considering that optimal information exchange patterns may differ from physical connectivity constraints. Recent research has explored communication optimization techniques, such as distributed event-triggered control \cite{najafirad1}, to reduce communication overhead in DC MGs, but these approaches do not account for the practical reality that communication topologies can vary over time due to the inherently intermittent nature of DGs. Moreover, recent advances in communication technologies, particularly software-defined radio (SDR) networks \cite{jin2017toward} and adaptive massive MIMO beamforming techniques \cite{chataut2020massive}, have eliminated the need for fixed communication structures. In fact, communication topology changes are often necessary to overcome large-scale disruptions, minimize disturbances, and address cybersecurity threats, making adaptive and reconfigurable communication networks essential for resilient MG operations.

Figure \ref{fig.initialresult} illustrates that the relationship between communication network complexity (number of links) and control performance in DC MGs exhibits a non-monotonic behavior. As shown in Fig. \ref{fig.initialresult}(a), increasing communication links initially improves both current sharing accuracy and voltage regulation performance but reaches an optimal threshold beyond which additional links provide diminishing returns \cite{zhou2020distributed,lou2018optimal}. The corresponding communication network topology in Fig. \ref{fig.initialresult}(b) demonstrates how the communication links are structured among DGs. This phenomenon occurs because excessive communication links can introduce unnecessary complexity, delays, and potential points of failure without corresponding performance benefits \cite{sheng2022optimal}. Therefore, identifying the minimal communication topology that achieves optimal voltage regulation and current sharing represents a critical design challenge that motivates the need for systematic co-design approaches that can simultaneously optimize both controller parameters and communication network structure.

This flexibility to co-design controllers and topology creates opportunities to develop innovative control strategies that can advantageously utilize customizable and reconfigurable communication topologies, leading to more robust and efficient MG operations. A cost-effective communication network while ensuring the performance of secondary control is addressed in \cite{hu2021cost} for MGs. In \cite{lou2018optimal}, a distributed secondary control is proposed to optimize the communication topology and controller parameters. The joint optimization problem of communication topology and weights is proposed in \cite{sheng2022optimal} to improve attack resilience and dynamic performance. In \cite{huang2024optimal}, the topology and weights of the communication network are designed for distributed control in discrete time to optimize the convergence performance of distributed secondary control in MG. However, their \cite{lou2018optimal,sheng2022optimal,hu2021cost,huang2024optimal} design procedure follows a sequential rather than a co-design strategy. Such a separation of design may result in suboptimal overall system performance compared to what could be achieved through co-design of both communication topology and controller parameters.

Passivity is a fundamental concept in control theory, offering a systematic approach to ensure stability in complex networked systems where traditional Lyapunov methods may be difficult to apply \cite{loranca2021data}. Passive systems maintain inherent stability even when interconnected with other passive systems through parallel or feedback connections \cite{hassan2020constant}. When a converter circuit becomes impassive, Passivity-based control (PBC) adjusts energy dissipation by adding a virtual resistance matrix that reduces oscillations and ensures global stability \cite{hassan2018adaptive}. This property enables a modular approach in which individual passive converters, supported by the PBC discipline and the self-disciplined stabilization nature of the interconnection strategy, ensure system stability without complex integrated network analysis \cite{gu2014passivity}. 

Moreover, dissipativity theory provides a more comprehensive framework that generalizes passivity by considering broader classes of supply rates and storage functions \cite{arcak2022}. A key advantage of dissipativity-based approaches is that they require only knowledge of the input-output energy relationships (in other words, dissipativity properties) of subsystems rather than their complete dynamic models \cite{welikala2023non}. This energy-based dissipativity abstraction significantly simplifies analysis and design, as identifying dissipativity properties is much more convenient, efficient, and reliable than developing detailed mathematical models of complex components \cite{arcak2022}. Furthermore, dissipativity conditions can be formulated as linear matrix inequalities (LMIs), allowing efficient numerical implementation and systematic controller synthesis \cite{welikala2023platoon}. The compositionality property of dissipative systems with compatible supply rates extends the compositionality benefits of passive systems due to greater flexibility to characterize subsystem behavior \cite{welikala2023non}. Due to these advantages, passivity and broader dissipativity concepts have gained popularity in industrial applications, particularly for power electronic converters and MGs \cite{kwasinski2010dynamic}.

In \cite{malan2024passivity}, a novel four-stage distributed controller is proposed that achieves power sharing and voltage regulation in DC MGs with both actuated and unactuated buses by exploiting dissipativity properties to ensure closed-loop stability across varying network topologies and actuation states. In \cite{nahata2020passivity}, a passivity-based approach to voltage stabilization in DC microgrids with ZIP loads provides explicit control gain inequalities that ensure system stability while enabling plug-and-play (PnP) operation. The authors of \cite{tucci2016decentralized} developed a decentralized voltage control method for DC MGs that enables PnP capability through local controllers requiring only knowledge of the corresponding DG's parameters. A distributed controller using PBC is proposed in \cite{cucuzzella2019distributed} to ensure proportional current sharing while maintaining voltage regulation in DC MGs with dynamic line models. In \cite{agarwal2020distributed}, a distributed synthesis method is developed for local controllers applicable to DC MGs where the addition of new components does not require a redesign of existing controllers. The line-independent PnP control approach is introduced in \cite{tucci2017line} for voltage restoration in DC MGs where local controllers can be designed independently of the power line parameters. In \cite{abdolahi2024passivity}, a combined PBC and nonlinear disturbance observer is developed for boost converters with constant power loads (CPLs) in DC MGs, achieving faster response times and better stability than conventional controllers. In \cite{malan2025passivation}, a PBC using LMIs is proposed to stabilize DC MGs with voltage-setting and voltage-following buses, ensuring asymptotic stability despite uncertain and nonmonotone loads.

This paper introduces a novel dissipativity-based distributed control framework for DC MGs that eliminates the need for traditional droop characteristics. Unlike existing droop-free distributed approaches \cite{dissanayake2019droop,zhang2022droop} that focus primarily on controller design with predetermined communication topologies, our method simultaneously optimizes both the distributed controller parameters and the communication network structure. Furthermore, while previous distributed control methods rely on traditional control synthesis techniques, our approach leverages dissipativity theory to provide systematic stability guarantees and enables a unified co-design framework. Our approach views DGs, loads, and transmission lines as interconnected energy systems and focuses on their fundamental energy exchanges. By applying dissipativity principles from system theory, we develop a comprehensive model that ensures both system stability and robust power sharing with disturbance rejection capabilities. Furthermore, we propose an innovative co-design methodology that simultaneously optimizes the distributed controllers and the communication network topology using LMIs. 

The main contributions of this paper can be outlined as:
\begin{enumerate}
\item
For the first time, we formulate the DC MG control problem as a networked system where the DGs, loads, and transmission lines are connected through an interconnection matrix, allowing systematic analysis and design of the entire MG.
\item 
We develop a novel co-design methodology that simultaneously optimizes local controllers, distributed controllers, and communication network topology to overcome sequential design limitations and achieve superior system performance.
\item 
We propose a droop-free control framework that combines dissipativity-based local voltage control with distributed current sharing to eliminate traditional droop mechanisms while ensuring both voltage regulation and proportional current sharing.
\item 
We formulate all design problems as LMI-based convex optimization problems to enable efficient numerical implementation and scalable controller synthesis with systematic stability guarantees.
\end{enumerate}

A preliminary version of this work was presented in \cite{najafi2025distributed}. This journal version extends the conference work by incorporating consensus-based distributed controllers for current sharing, error dynamics formulation for improved tracking performance, and comprehensive disturbance modeling for realistic operating conditions. This paper is structured as follows. Section \ref{Preliminaries} covers the necessary preliminary concepts on dissipativity and networked systems. Section \ref{problemformulation} presents the DC MG modeling and co-design problem formulation. The hierarchical control design is presented in Section \ref{Sec:ControlDesign}. 
The distributed control and topology co-design approach based on dissipativity is detailed in Section \ref{Passivity-based Control}. The simulation results are presented in Section \ref{Simulation}, followed by the conclusions in Section \ref{Conclusion}.

\section{Preliminaries}\label{Preliminaries}

\subsubsection*{\textbf{Notations}}
The notation $\mathbb{R}$ and $\mathbb{N}$ signify the sets of real and natural numbers, respectively. 
For any $N\in\mathbb{N}$, we define $\mathbb{N}_N\triangleq\{1,2,..,N\}$.
An $n \times m$ block matrix $A$ is denoted as $A = [A_{ij}]_{i \in \mathbb{N}_n, j \in \mathbb{N}_m}$. Either subscripts or superscripts are used for indexing purposes, e.g., $A_{ij} \equiv A^{ij}$.
$[A_{ij}]_{j\in\mathbb{N}_m}$ and $\diag([A_{ii}]_{i\in\mathbb{N}_n})$ represent a block row matrix and a block diagonal matrix, respectively.
$\0$ and $\I$, respectively, are the zero and identity matrices (dimensions will be clear from the context). A symmetric positive definite (semi-definite) matrix $A\in\mathbb{R}^{n\times n}$ is denoted by $A>0\ (A\geq0)$. The symbol $\star$ represents conjugate blocks inside block a symmetric matrices, $\mathcal{H}(A)\triangleq A + A^\T$ and $\mb{1}_{\{ \cdot \}}$ is the indicator function.

\subsection{Dissipativity}
Consider a general non-linear dynamic system
\begin{equation}\label{dynamic}
\begin{aligned}
    \dot{x}(t)=f(x(t),u(t)),\\
    y(t)=h(x(t),u(t)),
    \end{aligned}
\end{equation}
where $x(t)\in\mathbb{R}^n$, $u(t)\in\mathbb{R}^q$, $y(t)\in\mathbb{R}^m$, and $f:\mathbb{R}^n\times\mathbb{R}^q\rightarrow\mathbb{R}^n$ and $h:\mathbb{R}^n\times\mathbb{R}^q\rightarrow\mathbb{R}^m$ are continuously differentiable. 
The equilibrium points of (\ref{dynamic}) are such that there is a unique $u^*\in\mathbb{R}^q$ such that $f(x^*,u^*)=0$ for any $x^*\in\mathcal{X}$, where $\mathcal{X}\subset\mathbb{R}^n$ is the set of equilibrium states. And both $u^*$ and $y^*\triangleq h(x^*,u^*)$ are implicit functions of $x^*$.

The \textit{equilibrium-independent-dissipativity} (EID) \cite{arcak2022} is defined next to examine dissipativity of (\ref{dynamic}) without the explicit knowledge of its equilibrium points. 

\begin{definition}
The system \eqref{dynamic} is called EID under supply rate $s:\mathbb{R}^q\times\mathbb{R}^m\rightarrow\mathbb{R}$ if there is a continuously differentiable storage function $V:\mathbb{R}^n\times\mathcal{X}\rightarrow\mathbb{R}$ such that  $V(x,x^*)>0$ when $x\neq x^*$, $V(x^*,x^*)=0$, and 
\begin{center}
    $\dot{V}(x,x^*)=\nabla_xV(x,x^*)f(x,u)\leq s(u-u^*,y-y^*)$,
\end{center}
for all $(x,x^*,u)\in\mathbb{R}^n\times\mathcal{X}\times\mathbb{R}^q$.
\end{definition}

This EID property can be specialized based on the used supply rate $s(.,.)$.
The $X$-EID property, defined in the sequel, uses a quadratic supply rate determined by a coefficient matrix $X=X^\top\triangleq[X^{kl}]_{k,l\in\mathbb{N}_2}\in\mathbb{R}^{q+m}$ \cite{welikala2023non}.

\begin{definition}
The system (\ref{dynamic}) is $X$-EID if it is EID under the quadratic supply rate:
\begin{center}
$
s(u-u^*,y-y^*)\triangleq
\scriptsize
\begin{bmatrix}
    u-u^* \\ y-y^*
\end{bmatrix}^\top
\begin{bmatrix}
    X^{11} & X^{12}\\ X^{21} & X^{22}
\end{bmatrix}
\begin{bmatrix}
    u-u^* \\ y-y^*
\end{bmatrix}.
\normalsize
$
\end{center}
\end{definition}

\begin{remark}\label{Rm:X-DissipativityVersions}
If the system (\ref{dynamic}) is $X$-EID with:\\
1)\ $X = \scriptsize\begin{bmatrix}
    \0 & \frac{1}{2}\I \\ \frac{1}{2}\I & \0
\end{bmatrix}\normalsize$, then it is passive;\\
2)\ $X = \scriptsize\begin{bmatrix}
    -\nu\I & \frac{1}{2}\I \\ \frac{1}{2}\I & -\rho\I
\end{bmatrix}\normalsize$, then it is strictly passive ($\nu$ and $\rho$ are the input feedforward and output feedback passivity indices, denoted as IF-OFP($\nu,\rho$));\\
3)\ $X = \scriptsize\begin{bmatrix}
    \gamma^2\I & \0 \\ \0 & -\I
\end{bmatrix}\normalsize$, then it is $L_2$-stable ($\gamma$ is the $L_2$-gain, denoted as $L2G(\gamma)$);\\
in an equilibrium-independent manner (see also \cite{WelikalaP42022}). 
\end{remark}

If the system (\ref{dynamic}) is linear time-invariant (LTI), a necessary and sufficient condition for $X$-EID is provided in the following proposition as a linear matrix inequality (LMI) problem.

\begin{proposition}\label{Prop:linear_X-EID} \cite{welikala2023platoon}
The LTI system
\begin{equation*}\label{Eq:Prop:linear_X-EID_1}
\begin{aligned}
    \dot{x}(t)&=Ax(t)+Bu(t),\\
    y(t)&=Cx(t)+Du(t),
\end{aligned}
\end{equation*}
is $X$-EID if and only if there exists $P>0$ such that
\begin{equation*}\label{Eq:Prop:linear_X-EID_2}
\scriptsize
\begin{bmatrix}
-\mathcal{H}(PA)+C^\top X^{22}C & -PB+C^\top X^{21}+C^\top X^{22}D\\
\star & X^{11}+\mathcal{H}(X^{12}D)+D^\top X^{22}D
\end{bmatrix}
\normalsize
\geq0.
\end{equation*}
\end{proposition}

The following corollary considers a specific LTI system with a local controller $u(t)=Lx(t)$ (a setup that will be useful later) and formulates an LMI problem for $X$-EID enforcing local controller synthesis.
\begin{corollary}\label{Col.LTI_LocalController_XEID}\cite{welikala2023platoon}
The LTI system 
\begin{equation*}
    \dot{x}(t)=(A+BL)x(t)+\eta(t), \quad y(t) = x(t),
\end{equation*}
is $X$-EID with $X^{22}<0$ if and only if there exists $P>0$ and $K$ such that
\begin{equation*}
\scriptsize
    \begin{bmatrix}
        -(X^{22})^{-1} & P & 0\\
        \star & -\mathcal{H}(AP+BK) & -\I+PX^{21}\\
        \star & \star & X^{11}
    \end{bmatrix}\normalsize\geq0,
\end{equation*}
and $L=KP^{-1}$.
\end{corollary}

\subsection{Networked Systems} \label{SubSec:NetworkedSystemsPreliminaries}

Consider the networked system $\Sigma$ in Fig. \ref{Networked}, consisting of dynamic subsystems $\Sigma_i,i\in\mathbb{N}_N$, $\Bar{\Sigma}_i,i\in\mathbb{N}_{\Bar{N}}$ and a static interconnection matrix $M$ that characterizes interconnections among subsystems, exogenous inputs $w(t)\in\mathbb{R}^r$ (e.g. disturbances) and interested outputs $z(t)\in\mathbb{R}^l$ (e.g. performance). 

\begin{figure}
    \centering
    \includegraphics[width=0.6\columnwidth]{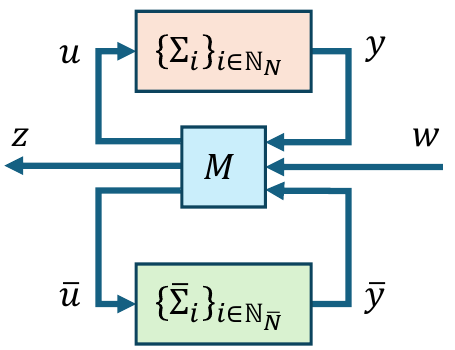}
    \caption{A generic networked system $\Sigma$.}
    \label{Networked}
\end{figure}

The dynamics of each subsystem $\Sigma_i,i\in\mathbb{N}_N$ are given by
\begin{equation}
    \begin{aligned}
        \dot{x}_i(t)&=f_i(x_i(t),u_i(t)),\\
        y_i(t)&=h_i(x_i(t),u_i(t)),
    \end{aligned}
\end{equation}
where $x_i(t)\in\mathbb{R}^{n_i}$, $u_i(t)\in\mathbb{R}^{q_i}$, $y_i(t)\in\mathbb{R}^{m_i}$. Similar to \eqref{dynamic}, each subsystem $\Sigma_i, i\in\N_N$ is considered to have a set $\mathcal{X}_i \subset \R^{n_i}$, where for every $x_i^* \in \mathcal{X}_i$, there exists a unique $u_i^* \in \R^{q_i}$ such that $f_i(x_i^*,u_i^*)=0$, and both $u_i^*$ and $y_i^*\triangleq h_i(x_i^*,u_i^*)$ are implicit function of $x_i^*$. Moreover, each subsystem $\Sigma_i, i\in\N_N$ is assumed to be $X_i$-EID, where $X_i \triangleq [X_i^{kl}]_{k,l\in\N_2}$. Regarding each subsystem $\bar{\Sigma}_i, i\in\N_{\bar{N}}$, we use similar assumptions and notations, but include a bar symbol to distinguish between the two types of subsystems, e.g.,  $\bar{\Sigma}_i$ is assumed to be $\bar{X}_i$-EID where $\bar{X}_i \triangleq [\bar{X}_i^{kl}]_{k,l\in\N_2}$.

Defining $u\triangleq[u_i^\top]^\top_{i\in\mathbb{N}_N}$, $y\triangleq[y_i^\top]^\top_{i\in\mathbb{N}_N}$, $\Bar{u}\triangleq[\Bar{u}_i^\top]^\top_{i\in\mathbb{N}_{\Bar{N}}}$ and $\bar{y}\triangleq[\bar{y}_i^\top]^\top_{i\in\mathbb{N}_{\bar{N}}}$, the interconnection matrix $M$ and the corresponding interconnection relationship are given by
\begin{equation}\label{interconnectionMatrix}
\scriptsize
\begin{bmatrix}
    u \\ \bar{u} \\ z
\end{bmatrix}=M
\normalsize
\scriptsize
\begin{bmatrix}
    y \\ \bar{y} \\ w
\end{bmatrix}
\normalsize
\equiv
\scriptsize
\begin{bmatrix}
    M_{uy} & M_{u\bar{y}} & M_{uw}\\
    M_{\bar{u}y} & M_{\bar{u}\bar{y}} & M_{\bar{u}w}\\
    M_{zy} & M_{z\bar{y}} & M_{zw}
\end{bmatrix}
\begin{bmatrix}
    y \\ \bar{y} \\ w
\end{bmatrix}.
\normalsize
\end{equation}

The following proposition exploits the $X_i$-EID and $\bar{X}_i$-EID properties of the subsystems $\Sigma_i,i\in\mathbb{N}_N$ and $\Bar{\Sigma}_i,i\in\mathbb{N}_{\Bar{N}}$ to formulate an LMI problem for synthesizing the interconnection matrix $M$ (\ref{interconnectionMatrix}), ensuring the networked system $\Sigma$ is $\textbf{Y}$-EID for a prespecified $\textbf{Y}$ under two mild assumptions \cite{welikala2023non}.

\begin{assumption}\label{As:NegativeDissipativity}
    For the networked system $\Sigma$, the provided \textbf{Y}-EID specification is such that $\textbf{Y}^{22}<0$.
\end{assumption}

\begin{remark}
Based on Rm. \ref{Rm:X-DissipativityVersions}, As. \ref{As:NegativeDissipativity} holds if the networked system $\Sigma$ must be either: (i) L2G($\gamma$) or (ii) IF-OFP($\nu,\rho$) with some $\rho>0$, i.e., $L_2$-stable or passive, respectively. Therefore, As. \ref{As:NegativeDissipativity} is mild since it is usually preferable to make the networked system $\Sigma$ either $L_2$-stable or passive.
\end{remark}

\begin{assumption}\label{As:PositiveDissipativity}
    In the networked system $\Sigma$, each subsystem $\Sigma_i$ is $X_i$-EID with $X_i^{11}>0, \forall i\in\mathbb{N}_N$, and similarly, each subsystem $\bar{\Sigma}_i$ is $\bar{X}_i$-EID with $\bar{X}_i^{11}>0, \forall i\in\N_{\bar{N}}$.
\end{assumption}

\begin{remark}
    According to Rm. \ref{Rm:X-DissipativityVersions}, As. \ref{As:PositiveDissipativity} holds if a subsystem $\Sigma_i,i\in\N_N$ is either: (i) L2G($\gamma_i$) or (ii) IF-OFP($\nu_i,\rho_i$) with $\nu_i<0$ (i.e., $L_2$-stable or non-passive). Since in passivity-based control, often the involved subsystems are non-passive (or can be treated as such), As. \ref{As:PositiveDissipativity} is also mild. 
\end{remark}

\begin{figure*}[!hb]
\vspace{-5mm}
\centering
\hrulefill
\begin{equation}\label{NSC3YEID}
\scriptsize \bm{
	\textbf{X}_p^{11} & \0 & L_{uy} & L_{u\bar{y}}\\ 
	\0 & \bar{\textbf{X}}_{\bar{p}}^{11} & L_{\bar{u}y} & L_{\bar{u}\bar{y}}\\
	L_{uy}^\T & L_{\bar{u}y}^\T & -L_{uy}^\T \textbf{X}^{12} - \textbf{X}^{21} L_{uy} -\textbf{X}_p^{22} & -\textbf{X}^{21}L_{u\bar{y}}-L_{\bar{u}y}^\T \bar{\textbf{X}}^{12} \\
	L_{u\bar{y}}^\T & L_{\bar{u}\bar{y}}^\T & -L_{u\bar{y}}^\T\textbf{X}^{12}-\bar{\textbf{X}}^{21} L_{\bar{u}y} & -L_{\bar{u}\bar{y}}^\T \bar{\textbf{X}}^{12} - \bar{\textbf{X}}^{21} L_{\bar{u}\bar{y}} - \bar{\textbf{X}}_{\bar{p}}^{22}
} \normalsize >0
\end{equation}
\end{figure*} 

\begin{proposition}\label{synthesizeM}\cite{welikala2023non}
    Under As. \ref{As:NegativeDissipativity}-\ref{As:PositiveDissipativity} (i.e., $\textbf{Y}^{22}<0$, $X_i^{11}>0$, and $\bar{X}_i^{11}>0$ for all subsystems), the network system $\Sigma$ can be made \textbf{Y}-EID (from $w(t)$ to $z(t)$) by synthesizing the interconnection matrix $M$ \eqref{interconnectionMatrix} via solving the LMI problem:
\begin{equation}
\begin{aligned}
	&\mbox{Find: } 
	L_{uy}, L_{u\bar{y}}, L_{uw}, L_{\bar{u}y}, L_{\bar{u}\bar{y}}, L_{\bar{u}w}, M_{zy}, M_{z\bar{y}}, M_{zw}, \\
	&\mbox{Sub. to: } p_i \geq 0, \forall i\in\N_N, \ \ 
	\bar{p}_l \geq 0, \forall l\in\N_{\bar{N}},\ \text{and} \  \eqref{NSC4YEID},
\end{aligned}
\end{equation}
with
\scriptsize
$\bm{M_{uy} & M_{u\bar{y}} & M_{uw} \\ M_{\bar{u}y} & M_{\bar{u}\bar{y}} & M_{\bar{u}w}} = 
\bm{\textbf{X}_p^{11} & \0 \\ \0 & \bar{\textbf{X}}_{\bar{p}}^{11}}^{-1} \hspace{-1mm} \bm{L_{uy} & L_{u\bar{y}} & L_{uw} \\ L_{\bar{u}y} & L_{\bar{u}\bar{y}} & L_{\bar{u}w}}$.
\normalsize
\end{proposition}


\begin{lemma}\label{Lm:Schur_comp}
\textbf{(Schur Complement)} For matrices $P > 0, Q$ and symmetric $R$, the following statements are equivalent \cite{boyd2004convex}:
\begin{subequations}
\begin{align}
1)\ &\begin{bmatrix} P & Q \\ Q^\T & R \end{bmatrix} \geq 0, \label{Lm:Schur1}\\
2)\ &R-Q^\T P^{-1}Q \geq 0. \label{Lm:Schur2}
\end{align}
\end{subequations}
\end{lemma}




\begin{figure*}[!hb]
\vspace{-5mm}
\centering
\begin{equation}\label{NSC4YEID}
\scriptsize
 \bm{
		\textbf{X}_p^{11} & \0 & \0 & L_{uy} & L_{u\bar{y}} & L_{uw} \\
		\0 & \bar{\textbf{X}}_{\bar{p}}^{11} & \0 & L_{\bar{u}y} & L_{\bar{u}\bar{y}} & L_{\bar{u}w}\\
		\0 & \0 & -\textbf{Y}^{22} & -\textbf{Y}^{22} M_{zy} & -\textbf{Y}^{22} M_{z\bar{y}} & \textbf{Y}^{22} M_{zw}\\
		L_{uy}^\T & L_{\bar{u}y}^\T & - M_{zy}^\T\textbf{Y}^{22} & -L_{uy}^\T\textbf{X}^{12}-\textbf{X}^{21}L_{uy}-\textbf{X}_p^{22} & -\textbf{X}^{21}L_{u\bar{y}}-L_{\bar{u}y}^\T \bar{\textbf{X}}^{12} & -\textbf{X}^{21}L_{uw} + M_{zy}^\T \textbf{Y}^{21} \\
		L_{u\bar{y}}^\T & L_{\bar{u}\bar{y}}^\T & - M_{z\bar{y}}^\T\textbf{Y}^{22} & -L_{u\bar{y}}^\T\textbf{X}^{12}-\bar{\textbf{X}}^{21}L_{\bar{u}y} & 		-(L_{\bar{u}\bar{y}}^\T \bar{\textbf{X}}^{12} + \bar{\textbf{X}}^{21}L_{\bar{u}\bar{y}}+\bar{\textbf{X}}_{\bar{p}}^{22}) & -\bar{\textbf{X}}^{21} L_{\bar{u}w} + M_{z\bar{y}}^\T \textbf{Y}^{21} \\ 
		L_{uw}^\T & L_{\bar{u}w}^\T & -M_{zw}^\T \textbf{Y}^{22}& -L_{uw}^\T\textbf{X}^{12}+\textbf{Y}^{12}M_{zy} & -L_{\bar{u}w}^\T\bar{\textbf{X}}^{12}+ \textbf{Y}^{12} M_{z\bar{y}} & M_{zw}^\T\textbf{Y}^{21} + \textbf{Y}^{12}M_{zw} + \textbf{Y}^{11}
	}\normalsize > 0
\end{equation}
\end{figure*}

\section{System Modeling}\label{problemformulation}
This section presents the dynamic modeling details of the DC MG, consisting of multiple DGs, loads, and transmission lines. Specifically, our modeling approach is motivated by \cite{nahata}, which highlights the role and impact of communication and physical topologies in DC MGs.

\subsection{DC MG Physical Interconnection Topology}
The physical interconnection topology of a DC MG is modeled as a directed connected graph $\mathcal{G}^p =(\mathcal{V},\mathcal{E})$ where $\mathcal{V} = \mathcal{D} \cup \mathcal{L}$ is bipartite: $\mathcal{D}=\{\Sigma_i^{DG}, i\in\N_N\}$ (DGs) and $\mathcal{L}=\{\Sigma_l^{line}, l\in\N_L\}$ (transmission lines). A representative diagram of a DC MG is illustrated in Fig. \ref{diagram}. The DGs are interconnected with each other through transmission lines. The interface between each DG and the DC MG is through a point of common coupling (PCC). For simplicity, the loads are assumed to be connected to the DG terminals at the respective PCCs \cite{dorfler2012kron}. Indeed loads can be moved to PCCs using Kron reduction even if they are located elsewhere \cite{dorfler2012kron}.

\begin{figure}
    \centering
    \includegraphics[width=0.5\columnwidth]{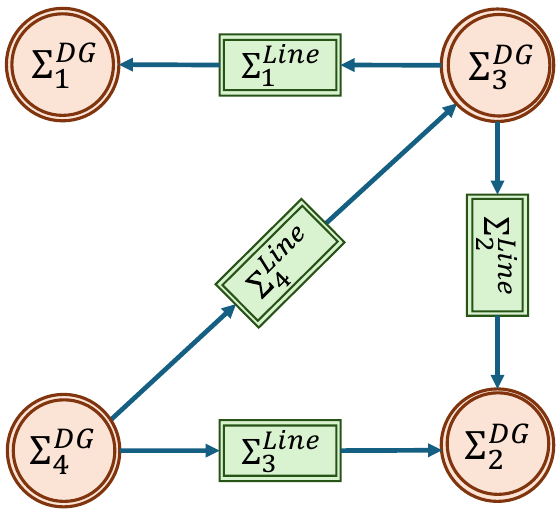}
    \caption{A simplified diagram of a DC MG network.}
    \label{diagram}
\end{figure}


To represent the DC MG's physical topology, we use its adjacency matrix $\mathcal{A} =  \scriptsize 
\begin{bmatrix}
\0 & \mathcal{B} \\
\mathcal{B}^\T & \0
\end{bmatrix},
\normalsize$  
where $\mathcal{B}\in\R^{N \times L}$ is the incident matrix of the DG network (where nodes are just the DGs and edges are just the transmission lines). Note that $\mathcal{B}$ is also known as the ``bi-adjacency'' matrix of $\mathcal{G}^p$ that describes the connectivity between its two types of nodes. In particular, $\mathcal{B}= [\mathcal{B}_{il}]_{i \in \N_N, l \in \N_L}$ with
$\mathcal{B}_{il} \triangleq \mb{1}_{\{l\in\mathcal{E}_i^+\}} - \mb{1}_{\{l\in\mathcal{E}_i^-\}},$ where $\mathcal{E}_i^+$ and $\mathcal{E}_i^-$ represent the out- and in-neighbors of $\Sigma_i^{DG}$. We define $\mathcal{E}_i = \mathcal{E}_i^+ \cup \mathcal{E}_i^-$ as the general set of all transmission lines connected to the $\Sigma_i^{DG}$.


\subsection{Dynamic Model of a Distributed Generator (DG)}
Each DG consists of a DC voltage source, a voltage source converter (VSC), and some RLC components, where each DG $\Sigma_i^{DG},i\in\N_N$ supplies power to a specific load in its PCC (denoted $\text{PCC}_i$). In addition, it interconnects with other DG units via transmission lines $\{\Sigma_l^{line}:l \in \mathcal{E}_i\}$. Figure \ref{DCMG} illustrates the schematic diagram of $\Sigma_i^{DG}$, including the local load, a connected transmission line, and the local and distributed global controllers (will be formally introduced in the sequel).

By applying Kirchhoff's Current Law (KCL) and Kirchhoff's Voltage Law (KVL) at $\text{PCC}_i$ on the DG side, we get the following equations for $\Sigma_i^{DG},i\in\N_N$:
\begin{equation}
\begin{aligned}\label{DGEQ}
\Sigma_i^{DG}:
\begin{cases}
    C_{ti}\frac{dV_i}{dt} &= I_{ti} - I_{Li}(V_i) - I_i + w_{vi}(t), \\
L_{ti}\frac{dI_{ti}}{dt} &= -V_i - R_{ti}I_{ti} + V_{ti} + w_{ci}(t),
\end{cases}
\end{aligned}
\end{equation}
where the parameters $R_{ti}$, $L_{ti}$, and $C_{ti}$ 
represent the internal resistance, internal inductance, and filter capacitance of $\Sigma_i^{DG}$, respectively. The state variables are selected as $V_i$ and $I_{ti}$, where $V_i$ is the $\text{PCC}_i$ voltage and $I_{ti}$ is the internal current. Without loss of generality, $w_{vi}(t)$ and $w_{ci}(t)$ in \eqref{DGEQ} represent unknown disturbances affecting the voltage and current dynamics, assumed to be zero mean Gaussian distributed with variances $\sigma_{vi}^2$ and $\sigma_{ci}^2$, respectively. 
In addition, $V_{ti}$ is the input command signal applied to the VSC, $I_{Li}(V_i)$ is the load current, and $I_i$ is the total current injected into the DC MG by $\Sigma_i^{DG}$, where $V_{ti}$, $I_{Li}(V_i)$, and $I_i$ are respectively determined by the controllers, loads, and lines at $\Sigma_i^{DG}$. The total line current $I_i$ is given by 
\begin{equation}
\label{Eq:DGCurrentNetOut}
I_i
=\sum_{l\in\mathcal{E}_i}\mathcal{B}_{il}I_l,
\end{equation}
where $I_l, l\in\mathcal{E}_i$ are line currents.

\subsection{Dynamic Model of a Transmission Line}
Each transmission line is modeled using the $\pi$-equivalent representation, where we assume that the line capacitances are consolidated with the capacitances of the DG filters. Consequently, as shown in Fig. \ref{DCMG}, the power line $\Sigma_l^{line}$ can be represented as an RL circuit with resistance $R_l$ and inductance $L_l$. By applying KVL to $\Sigma_l^{line}$, we obtain:
\begin{equation}\label{line}
    \Sigma_l^{line}:
        L_l\frac{dI_l}{dt}=-R_lI_l+\Bar{u}_l + \bar{w}_l(t),
\end{equation}
where $I_l$ is the line current state and $\Bar{u}_l=V_i-V_j=\sum_{i\in \mathcal{E}_l}\mathcal{B}_{il}V_i$, and $\bar{w}_l(t)$ represents an unknown disturbance that affects the line dynamics assumed to be a zero-mean Gaussian distributed with variance $\sigma_l^2$. This disturbance term captures the effect of the uncertainty in line resistance. 

\begin{figure*}
    \centering
    \includegraphics[width=2\columnwidth]{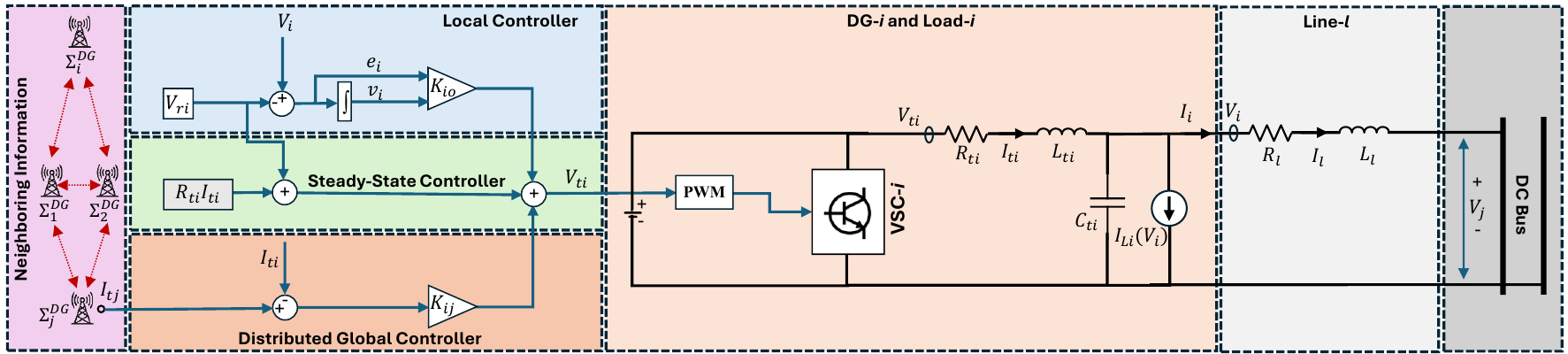}
    \caption{The electrical schematic of DG-$i$, load-$i$, $i\in\N_N$, local controller, steady-state controller, distributed global controller, and line-$l$, $l\in\N_L$.}
    \label{DCMG}
\end{figure*}

\subsection{Dynamic Model of a  Load} 
Recall that $I_{Li}(V_i)$ in \eqref{DGEQ} and Fig. \ref{DCMG} is the current flowing through the load at $\Sigma_i^{DG}, i\in\N_N$. In the most general case \cite{kundur2007power}, the load can be thought of as a ``ZIP'' load, where $I_{Li}(V_i)$ takes the form:
\begin{equation}\label{Eq:LoadModel}
I_{Li}(V_i) = I_{Li}^Z(V_i) + I_{Li}^I(V_i) + I_{Li}^P(V_i).
\end{equation}

Here, the components of the ZIP load are:
(i) a constant impedance load: $I_{Li}^{Z}(V_i)=Y_{Li}V_i$, where  $Y_{Li}=1/R_{Li}$ is the conductance of the load,
(ii) a constant current load: $I_{Li}^{I}(V_i)=\bar{I}_{Li}$ where $\bar{I}_{Li}$ is the current demand of the load, and
(iii) a constant power load: $I_{Li}^{P}(V_i)=V_i^{-1}P_{Li}$, where $P_{Li}$ represents the power demand of the load. The control strategies for regulating the voltage and current of each DG will be detailed in the following section.

\section{Proposed Controller}
The primary objective is to ensure that the $\text{PCC}_i$ voltage $V_i$ at each $\Sigma_i^{DG},i\in\N_N $ closely follows a specified reference voltage $V_{ri}$ while achieving proportional current sharing among DGs. Local PI controllers handle voltage regulation at each $\Sigma_i^{DG}$, while consensus-based distributed global controllers ensures proper current sharing across the MG.

\subsection{Local Voltage Controller}
At each $\Sigma_i^{DG}, i\in\N_N$, to effectively track the assigned reference voltage $V_{ri}$, it is imperative to ensure that the error $e_i(t)\triangleq V_i(t)-V_{ri}$ converges to zero, i.e. $\lim_{t \to \infty} (V_i(t) - V_{ri}) = 0$, which establishes voltage regulation.
To this end, motivated by \cite{tucci2017}, we first include each $\Sigma_i^{DG}, i\in\N_N$ with an integrator state $v_i$ (i.e. $v_i(t) = \int_0^t (V_i(t) - V_{ri}) \ dt$) (see Fig. \ref{DCMG}) that follows the dynamics 
\begin{equation}\label{error}
    \frac{dv_i}{dt}=e_i(t)=V_i(t)-V_{ri}.
\end{equation}
 Then, $\Sigma_i^{DG}$ is equipped with a local state feedback controller:
 
\begin{equation}\label{Controller}
  u_{iL}(t)\triangleq k_{i0}^P e_i(t)+k_{i0}^I \int_0^t e_i(t) \ dt 
  = K_{i0}x_i(t) - k_{i0}^P V_{ri},
\end{equation}
where
\begin{equation}\label{Eq:DGstate}
x_i \triangleq \begin{bmatrix}
    V_i &  I_{ti} & v_i
\end{bmatrix}^\top,
\end{equation}
denotes the augmented state of $\Sigma_i^{DG}$ and $K_{i0}=\begin{bmatrix}
    k_{i0}^P & 0 & k_{i0}^I
\end{bmatrix}\in\mathbb{R}^{1\times3}$ where $K_{i0}$ is the local controller gain, and $k_{i0}^P$ and $k_{i0}^I$ are the proportional and integral local controller gains, respectively.

\subsection{Distributed Global Controller}
We implement a distributed global controller explicitly focused on current sharing that uses a consensus-based approach to ensure proportional current sharing according to each unit's capacity. The objective is to achieve precise current sharing such that:
\begin{equation}\label{Eq:current_sharing}
\frac{I_{ti}(t)}{P_{ni}} = \frac{I_{tj}(t)}{P_{nj}} = I_s, \quad \forall i,j\in\mathbb{N}_N,
\end{equation}
where $P_{ni}$ and $P_{nj}$ represent the power ratings of $\Sigma_i^{DG}$ and $\Sigma_j^{DG}$ respectively, and $I_s$ represents the common current sharing ratio that emerges from balancing the total load demand among DGs according to their power ratings.

The current sharing index $I_s$ is determined by balancing the total system load with the available DG capacities, ensuring that the aggregate power demand does not exceed the collective generation capability of all DGs. The specific value of $I_s$ depends on the system's load conditions, DG power ratings, and desired operating margins. As we will show in Th. \ref{Th:VRegulation_CSharing}, the optimal selection of $I_s$ along with the reference voltage vector $V_r$ can be formulated as an optimization problem that ensures both feasible current sharing and voltage regulation objectives are satisfied.

To address the current sharing motivated by \eqref{Eq:current_sharing}, as shown in Fig. \ref{DCMG}, we employ a consensus-based distributed controller of the form
\begin{equation}\label{ControllerG}
    u_{iG}(t)\triangleq  \sum_{j\in\bar{\mathcal{F}}_i^-} k_{ij}^I\left(\frac{I_{tj}(t)}{P_{nj}} - \frac{I_{ti}(t)}{P_{ni}}\right),
\end{equation} 
where each $k_{ij}^I\in\mathbb{R}$ is a distributed controller gain.

To implement this distributed controller, we denote the communication topology as a directed graph $\mathcal{G}^c =(\mathcal{D},\mathcal{F})$ where $\mathcal{D}\triangleq\{\Sigma_i^{DG}, i\in\N_N\}$ and $\mathcal{F}$ represents the set of communication links among DGs. The notations $\mathcal{F}_i^+$ and $\mathcal{F}_i^-$ are defined as the communication-wise out- and in-neighbors, respectively.

Now, the overall control input $u_i(t)$ applied to the VSC of $\Sigma_i^{DG}$ (see \eqref{DGEQ}) can be expressed as:
\begin{equation}\label{controlinput}
    u_i(t) \triangleq V_{ti}(t) =  u_{iS} + u_{iL}(t) + u_{iG}(t),
\end{equation}
where $u_{iL}$ is given by \eqref{Controller}, $u_{iG}$ is given by \eqref{ControllerG}. The term $u_{iS}$ represents the steady-state control input, which, as we will see in the sequel, plays a crucial role in achieving the desired equilibrium point of the DC MG. In particular, this steady-state component ensures that the system can maintain its operating point while satisfying both voltage regulation and current sharing objectives through the sharing coefficient $I_s$. The specific structure and properties of $u_{iS}$ will be characterized through our stability analysis in Sec. \ref{Sec:Equ_Analysis}. 


\subsection{Closed-Loop Dynamics of the DC MG}

By combining \eqref{DGEQ} and \eqref{error}, the overall dynamics of $\Sigma_i^{DG}, i\in\N_N$ can be written as
\begin{subequations}\label{statespacemodel}
\begin{align}
       \frac{dV_i}{dt}&=\frac{1}{C_{ti}}I_{ti}-\frac{1}{C_{ti}}I_{Li}(V_i)-\frac{1}{C_{ti}}I_i+\frac{1}{C_{ti}}w_{vi}(t), \label{Eq:ss:voltages}\\
        \frac{dI_{ti}}{dt}&=-\frac{1}{L_{ti}}V_i-\frac{R_{ti}}{L_{ti}}I_{ti}+\frac{1}{L_{ti}}u_i + \frac{1}{L_{ti}}w_{ci}(t), \label{Eq:ss:currents}\\
        \frac{dv_i}{dt}&=V_i-V_{ri} \label{Eq:ss:ints},
\end{align}
\end{subequations}
where the terms  $I_i$, $I_{Li}(V_i)$, and $u_i$ can all be substituted from Eqs. \eqref{Eq:DGCurrentNetOut}, \eqref{Eq:LoadModel}, and \eqref{controlinput}, respectively.
Together with the commonly adopted assumption $P_{Li}\triangleq0$ (no constant power load) \cite{tucci2017}, we can restate \eqref{statespacemodel} as
\begin{equation}
\label{Eq:DGCompact}
\dot{x}_i(t)= A_ix_i(t)+B_iu_i(t)+E_id_i(t)+\xi_i(t),
\end{equation}
where $x_i(t)$ is the DG state as defined in \eqref{Eq:DGstate}, $d_i(t)$ is the exogenous input (disturbance) defined as 
\begin{equation}
d_i(t) \triangleq 
\bar{w}_i + w_i(t),
\end{equation}
with 
$\bar{w}_i \triangleq \bm{
    -\Bar{I}_{Li}  & 0 & -V_{ri}
}^\T$ representing the known fixed (mean) disturbance and
$w_i(t) \triangleq \bm{w_{vi}(t) & w_{ci}(t) & 0}^\T$ representing the unknown zero-mean disturbance. 
In \eqref{Eq:DGCompact}, $E_i \triangleq \diag(\bm{C_{ti}^{-1} & L_{ti}^{-1} & 1})$ is the disturbance input matrix, $\xi_i \triangleq \begin{bmatrix}
    -C_{ti}^{-1}\sum_{l\in \mathcal{E}_i} \mathcal{B}_{il}I_l & 0 & 0
\end{bmatrix}^\top$ is the transmission line coupling, and the remaining system matrices $A_i$, $B_i$, respectively, are
\begin{equation}\label{Eq:DG_Matrix_definition}
A_i \triangleq 
\begin{bmatrix}
  -\frac{Y_{Li}}{C_{ti}} & \frac{1}{C_{ti}} & 0\\
-\frac{1}{L_{ti}} & -\frac{R_{ti}}{L_{ti}} & 0 \\
1 & 0 & 0
\end{bmatrix}\ 
\mbox{ and }\ 
B_i \triangleq
\begin{bmatrix}
 0 \\ \frac{1}{L_{ti}} \\ 0
\end{bmatrix}.
\end{equation}
 
Similarly, using \eqref{line}, the state space representation of the transmission line $\Sigma_l^{Line}$ can be written in a compact form:
\begin{equation}\label{Eq:LineCompact}
    \dot{\bar{x}}_l(t) = \bar{A}_l\bar{x}_l(t) + \bar{B}_l\bar{u}_l + \bar{E}_l\bar{w}_l(t),
 \end{equation}
where $\bar{x}_l \triangleq I_l$ is the transmission line state,  $\bar{u}_l \triangleq \sum_{i\in \mathcal{E}_l}\mathcal{B}_{il}V_i$ is the voltage difference across the transmission line, $\bar{w}_l(t)$ is unknown zero-mean disturbance.
The disturbance input matrix is $\bar{E}_l \triangleq \bm{\frac{1}{L_{l}}}$ and the system matrices $\bar{A}_l$ and $\Bar{B}_l$ in \eqref{Eq:LineCompact} respectively, are 
\begin{equation}\label{Eq:LineMatrices}
\bar{A}_l \triangleq 
\begin{bmatrix}
-\frac{R_l}{L_l}
\end{bmatrix}\ 
\mbox{ and }\ 
\Bar{B}_l \triangleq  \begin{bmatrix}
\frac{1}{L_l}
\end{bmatrix}.  
\end{equation}

\subsection{Networked System Model}\label{Networked System Model}
Let us define $u\triangleq[u_i]_{i\in\N_N}$ and $\Bar{u}\triangleq[\Bar{u}_l]_{l\in\N_L}$ as vectorized control inputs, $x\triangleq[x_i]_{i\in\N_N}$ and $\Bar{x}\triangleq[\Bar{x}_l]_{l\in\N_L}$ as the full states, and $w\triangleq[w_i]_{i\in\N_N}$ and $\bar{w}\triangleq[\bar{w}_l]_{l\in\N_L}$ as disturbance inputs of DGs and lines from \eqref{Eq:DGCompact} and \eqref{Eq:LineCompact}, respectively.

Using these notations, we can now represent the DC MG as two sets of subsystems (i.e., DGs and lines) interconnected with disturbance inputs through a static interconnection matrix $M$ as shown in Fig. \ref{netwoked}. From comparing Fig. \ref{netwoked} with Fig. \ref{Networked}, it is clear that the DC MG takes the form of a standard networked system discussed in Sec. \ref{SubSec:NetworkedSystemsPreliminaries}, (besides a few notable differences which we will handle in the sequel). To identify the specific structure of the interconnection matrix $M$, we need to closely observe how the dynamics of DGs and lines are interconnected with each other, and their coupling with disturbance inputs.  

\begin{figure}
    \centering
    \includegraphics[width=0.9\columnwidth]{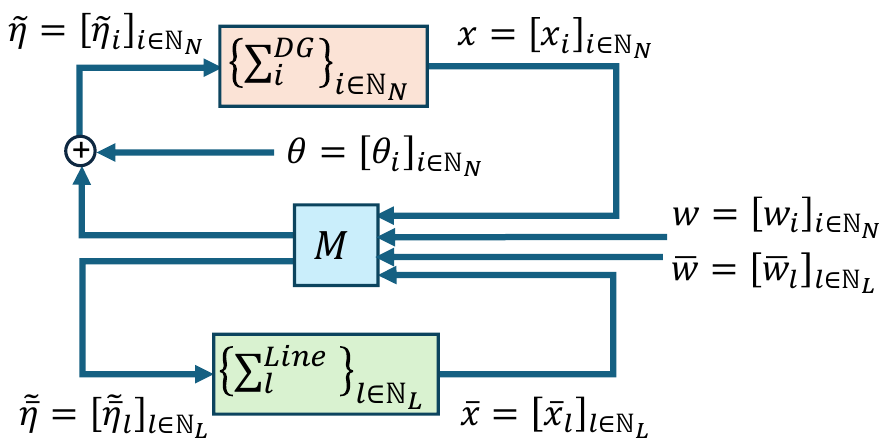}
    \caption{DC MG dynamics as a networked system configuration.}
    \label{netwoked}
\end{figure}

To this end, we first use \eqref{Eq:DGCompact} and \eqref{controlinput} to state the closed-loop dynamics of $\Sigma_i^{DG}$ as (see also Co. \ref{Col.LTI_LocalController_XEID})
\begin{equation}\label{closedloopdynamic}
    \dot{x}_i=(A_i+B_iK_{i0})x_i+\tilde{\eta}_i,
\end{equation}
where 
$\tilde{\eta}_i$ takes the form 
\begin{equation}
\label{Eq:DGClosedLoopDynamics_varphi}
\tilde{\eta}_i \triangleq \sum_{j\in\bar{\mathcal{F}}_i^-}K_{ij}x_j + \sum_{l\in\mathcal{E}_i}\Bar{C}_{il}\Bar{x}_l + E_iw_i(t) + \theta_i,
\end{equation}
with
$\Bar{C}_{il} \triangleq -C_{ti}^{-1}\begin{bmatrix}
\mathcal{B}_{il} & 0 & 0
\end{bmatrix}^\top, \forall l\in\mathcal{E}_i$, $\theta_i \triangleq  E_i\bar{w}_i + B_iu_{iS} - B_ik_{i0}^P V_{ri}$, and
\begin{equation}\label{k_ij}
    K_{ii} \triangleq
    \frac{1}{L_{ti}}
    \bm{
        0 & 0 & 0\\
        0 & \frac{\sum_{j\in\mathcal{F}_i^-} k_{ij}^I}{P_{ni}} & 0 \\
        0 & 0 & 0
    },\
        K_{ij} \triangleq
    \frac{1}{L_{ti}}
    \bm{
        0 & 0 & 0 \\
        0 & \frac{-k_{ij}^I}{P_{nj}} & 0 \\
        0 & 0 & 0    
    }.
\end{equation}

From \eqref{k_ij}, observe that only the (2,2)-th element in each block $K_{ij}$ is non-zero. Let $K_I \in \R^{N \times N}$ denote the matrix containing only these (2,2) block entries, i.e., $K_I = [K_{ij}^{2,2}]_{i,j\in\N_N}$. The controller gain matrix satisfies the weighted Laplacian property:
\begin{equation}
    K_I P_n \textbf{1}_N = 0,
\end{equation}
where $P_n = \diag(\bm{P_{ni}}_{i\in\N_N})$ and $\textbf{1}_N \in \R^N$ is the vector of ones. This ensures that the distributed control vanishes when proportional current sharing is achieved among all DGs.

By vectorizing \eqref{Eq:DGClosedLoopDynamics_varphi} over all $i\in\N_N$, we get 
\begin{equation}\label{Eq:DGClosedLoopInputVector}
    \tilde{\eta} \triangleq  Kx+\Bar{C}\Bar{x}+Ew+\theta,
\end{equation}
where $\tilde{\eta} \triangleq [\tilde{\eta}_i]_{i\in \N_N}$ represents the effective input vector to the DGs, $\Bar{C}\triangleq[\Bar{C}_{il}]_{i\in\mathbb{N}_N,l\in\mathbb{N}_L}$, $K \triangleq [K_{ij}]_{i,j\in\mathbb{N}_N}$, $E \triangleq \diag([E_i]_{i\in\N_N})$, and $\theta\triangleq [\theta_i]_{i\in\N_N}$.

\begin{remark}
The block matrices $K$ and $\Bar{C}$ in \eqref{Eq:DGClosedLoopInputVector} are indicative of the communication and physical topologies of the DC MG, respectively. The coupling gain matrix $K_{ij}$ is designed 
to capture the distributed current sharing objective through the communication network, where its structure reflects that only the current states are coupled via consensus gain $k_{ij}^I$ between $\Sigma_i^{DG}$ and $\Sigma_j^{DG}$. In particular, the $(i,j)$\tsup{th} block in $K$, i.e., $K_{ij}$ indicates a communication link from $\Sigma_j^{DG}$ to $\Sigma_i^{DG}$. Similarly, the $(i,l)$\tsup{th} block in $\bar{C}$ indicates a physical connection between $\Sigma_i^{DG}$ and $\Sigma_l^{Line}$.
\end{remark}


Following similar steps for transmission lines, using \eqref{Eq:LineCompact}, we can write
\begin{equation}
    \dot{\bar{x}}_l = \bar{A}_l\bar{x}_l + \tilde{\bar{\eta}}_l,
\end{equation}
where $\tilde{\bar{\eta}}$ can be stated as
\begin{equation}\label{Eq:linecontroller}
    \tilde{\Bar{\eta}}_l=\sum_{i\in\mathcal{E}_l}C_{il}x_i + \bar{E}_l\bar{w}_l(t),
\end{equation}
with $C_{il}\triangleq \begin{bmatrix}
    \mathcal{B}_{il} & 0 & 0
\end{bmatrix}, \forall l\in\mathcal{E}_i$. 
Note also that $C_{il} = -C_{ti}\bar{C}_{il}^\T$.
By vectorizing \eqref{Eq:linecontroller} over all $l\in\N_L$, we get
\begin{equation}\label{ubar}
    \tilde{\Bar{\eta}}=Cx + \bar{E}\bar{w},
\end{equation}
where $\tilde{\bar{\eta}} \triangleq [\tilde{\bar{\eta}}_l]_{l\in\N_L}$ represents the effective input vector to the lines, $C\triangleq[C_{il}]_{l\in\mathbb{N}_L,i\in\mathbb{N}_N}$, $\bar{E} \triangleq \diag([\bar{E}_l]_{l\in\N_L})$, and $\bar{w} \triangleq [\bar{w}_l]_{l\in\N_L}$. Note also that $C = - \bar{C}^\T C_t$ where $C_t \triangleq \diag([C_{ti}\I_{3}]_{i\in\N_N})$.





Finally, using \eqref{Eq:DGClosedLoopInputVector} and \eqref{ubar}, we can identify the interconnection relationship: 
\begin{equation}\nonumber
\begin{bmatrix}
\tilde{\eta} \\
\tilde{\bar{\eta}}
\end{bmatrix}
= M
\begin{bmatrix}
x \\
\bar{x} \\
w \\
\bar{w}
\end{bmatrix},
\end{equation}
where the interconnection matrix $M$ takes the form:
\begin{equation}\label{Eq:MMatrix}
M \triangleq 
\begin{bmatrix}
    K & \Bar{C} & E & \0 \\
    C & \0 & \0 & \bar{E}\\
\end{bmatrix}.
\end{equation}

When the physical topology $\mathcal{G}^p$ is predefined, so are the block matrices $\Bar{C}$ and $C$ (recall $C = -\bar{C}^\T C_t$). This leaves only the block matrix $K$ inside the block matrix $M$ as a tunable quantity to optimize the desired properties of the closed-loop DC MG system. Note that synthesizing $K$ simultaneously determines the distributed global controllers and the communication topology $\mathcal{G}^c$. In the following two sections, we provide a systematic approach to synthesize this block matrix $K$ to enforce stability and dissipativity, respectively.

\section{Networked Error Dynamics}\label{Sec:ControlDesign}
This section analyzes the equilibrium points and error dynamics of the DC MG system under the proposed hierarchical control structure. We examine the stability challenges where voltage regulation and current sharing objectives may conflict, and establish the mathematical foundation for the dissipativity-based control synthesis. The analysis focuses on deriving the networked error system representation and identifying the conditions necessary for system stability and performance.

\subsection{Equilibrium Point Analysis of the DC MG}\label{Sec:Equ_Analysis}
\begin{lemma}\label{Lm:equilibrium}
Assuming all external disturbances to be zero, i.e., $w_i(t)=0, \forall i\in\N_N$ and $\bar{w}_l(t)=0,\forall l\in\N_L$, for a given reference voltage vector $V_r$, constant current load vector $\bar{I}_L$, under a fixed DG control input $u_E$ defined as
\begin{equation}
    u_E \triangleq [\I + R_t(\mathcal{B}R^{-1}\mathcal{B}^\T + Y_L)]V_r + R_t \bar{I}_L, 
\end{equation}
there exists an equilibrium point for the DC MG given by:
\begin{equation}\label{Eq:Equilibrium}
\begin{aligned}
V_E &= V_{r},\\
I_{tE} &= (\mathcal{B}R^{-1}\mathcal{B}^\T + Y_L)V_r + \bar{I}_L,\\
\bar{I}_E &= R^{-1}\mathcal{B}^\T V_r,
\end{aligned}
\end{equation}
where we define the equilibrium state vectors $V_E\triangleq[V_{iE}]_{i\in\N_N}$,
$I_{tE}\triangleq[I_{tiE}]_{i\in\N_N}$,
$\bar{I}_E\triangleq[\bar{I}_{lE}]_{l\in\N_L}$ with 
$u_E\triangleq[u_{iE}]_{i\in\N_N}$ and the system parameters $C_t\triangleq\diag([C_{ti}]_{i\in\N_N})$,
$Y_L\triangleq\diag([Y_{Li}]_{i\in\N_N})$,
$L_t\triangleq\diag([L_{ti}]_{i\in\N_N})$,
$R_t\triangleq\diag([R_{ti}]_{i\in\N_N})$,
$\bar{I}_L\triangleq[\bar{I}_{Li}]_{i\in\N_N}$,
$V_r\triangleq[V_{ri}]_{i\in\N_N}$,
$R\triangleq\diag([R_l]_{l\in\mathbb{N}_L})$,
$\mathcal{B}\triangleq [\mathcal{B}_{il}]_{i\in\N_N,l\in\mathbb{N}_L}$.
\end{lemma}

\begin{proof}
The equilibrium state of the closed-loop $\Sigma_i^{DG}$ dynamics \eqref{Eq:DGCompact} satisfies:
\begin{equation}
\label{Eq:EqDGCompact}
A_ix_{iE}(t)+B_iu_{iE}(t)+E_id_{iE}(t)+\xi_{iE}(t) = 0,
\end{equation}
where $x_{iE} \triangleq \bm{V_{iE} & I_{tiE} & v_{iE}}^\T$ represents the equilibrium state components of DG, and $d_{iE}\triangleq \bar{w}_i$ and $\xi_{iE} \triangleq \begin{bmatrix}
    -C_{ti}^{-1}\sum_{l\in \mathcal{E}_i} \mathcal{B}_{il}\bar{I}_{lE} & 0 & 0
\end{bmatrix}^\top$ represent the equilibrium values of disturbance and interconnection terms, respectively. Thus, we get
\begin{equation}
\begin{aligned}
\begin{bmatrix} -\frac{Y_{Li}}{C_{ti}} & -\frac{1}{C_{ti}} & 0\\ -\frac{1}{L_{ti}} & -\frac{R_{ti}}{L_{ti}} & 0 \\ 1 & 0 & 0\end{bmatrix} \begin{bmatrix} V_{iE} \\ I_{tiE} \\ v_{iE} \end{bmatrix} + \begin{bmatrix} 0 \\ \frac{1}{L_{ti}} \\ 0 \end{bmatrix}u_{iE} + E_i\bar{w}_i + \xi_{iE} = 0.
\end{aligned}
\end{equation}

From the first row of this matrix equation, we get:
\begin{align}
\label{Eq:steadystatevoltage}
&-\frac{Y_{Li}}{C_{ti}}V_{iE} + \frac{1}{C_{ti}}I_{tiE} - \frac{1}{C_{ti}}\bar{I}_{Li} - \frac{1}{C_{ti}}\sum_{l\in \mathcal{E}i} \mathcal{B}_{il}\bar{I}_{lE} = 0.
\end{align}

From the last two rows of this matrix equation, we get
\begin{align}
\label{Eq:controlEquil}
    &u_{iE} = V_{iE} + R_{ti}I_{tiE},\\
 \label{Eq:voltageEquil}
    &V_{iE} = V_{ri}.
\end{align}

To further simplify the first equation, we require an expression for $\bar{I}_{lE}$. Note that the equilibrium state of the $\Sigma_l^{Line}$ \eqref{Eq:LineCompact} satisfies:
\begin{equation}\label{Eq:EqLineCompact}
    \bar{A}_l\bar{x}_{lE}(t) + \bar{B}_l\bar{u}_{lE} + \bar{E}_l\bar{w}_{lE}(t) = 0,
 \end{equation}
where $\bar{x}_{lE} \triangleq \bar{I}_{lE}$ represents the equilibrium state of line. The $\bar{u}_{lE}\triangleq \sum_{i\in \mathcal{E}_l}\mathcal{B}_{il}V_{iE}$ and $\bar{w}_{lE} \triangleq 0$ represent the equilibrium values of control input and disturbance of lines, respectively. Therefore, from \eqref{Eq:EqLineCompact} we get:
\begin{equation}
\label{Eq:LineEquil}
    \bar{I}_{lE} = \frac{1}{R_l}\sum_{i\in \mathcal{E}_l}\mathcal{B}_{il}V_{iE} = \frac{1}{R_l}\sum_{j\in \mathcal{E}_l}\mathcal{B}_{jl}V_{jE},
\end{equation}
which now can be applied in \eqref{Eq:steadystatevoltage} (together with \eqref{Eq:voltageEquil}) to obtain
\begin{equation}\label{Eq:currentEquil}
    -Y_{Li}V_{ri} + I_{tiE} - \sum_{l\in\mathcal{E}_i}\mathcal{B}_{il}\bigg(\frac{1}{R_l}\sum_{j\in \mathcal{E}_l}\mathcal{B}_{jl}V_{rj}\bigg)- \bar{I}_{Li} = 0.
\end{equation}

We next vectorize these equilibrium conditions \eqref{Eq:controlEquil}, \eqref{Eq:LineEquil}, \eqref{Eq:currentEquil}. Vectorizing \eqref{Eq:currentEquil} over all $i\in\N_N$ gives:
\begin{equation}\nonumber
    -Y_L V_r + I_{tE} - \mathcal{B}R^{-1}\mathcal{B}^\T V_r - \bar{I}_L = 0,
\end{equation}
leading to
\begin{equation}\nonumber
    I_{tE} = (\mathcal{B}R^{-1}\mathcal{B}^\T + Y_L)V_r + \bar{I}_L.
\end{equation}

For the equilibrium line currents, by vectorizing \eqref{Eq:LineEquil}, we get 
\begin{equation}\nonumber
    \bar{I}_E = R^{-1}\mathcal{B}^\T V_r.
\end{equation}

Therefore, the vectorized control equilibrium equation \eqref{Eq:controlEquil} can be expressed as:
\begin{equation}
\begin{aligned}\nonumber
    u_E &= V_r + R_t I_{tE} \\
    &= V_r + R_t((\mathcal{B}R^{-1}\mathcal{B}^\T + Y_L)V_r + \bar{I}_L\\
    &= [\I + R_t(\mathcal{B}R^{-1}\mathcal{B}^\T + Y_L)]V_r + R_t\bar{I}_L.
    \end{aligned}
\end{equation}

This completes the proof, as we have proved all the equilibrium conditions given in the Lemma statement. 
\end{proof}

\begin{lemma}\label{Lm:existence_uniqueness}
For any given reference voltage vector $V_r \in \mathbb{R}^N$ and constant current load vector $\bar{I}_L \in \mathbb{R}^N$, the equilibrium state variables $(V_E, I_{tE}, \bar{I}_E)$ and the control input $u_E$ in Lm. \ref{Lm:equilibrium} exist and are unique if and only if the system matrix in the equilibrium equation is invertible.
\end{lemma}
\begin{proof}
From the dynamics of the DC MG, at equilibrium $\dot{x} = 0$, the system matrix corresponds to $(\mathcal{B}R^{-1}\mathcal{B}^{\top} + Y_L)$ in our formulation. The equilibrium state variables are given by:
\begin{align*}
V_E &= V_r, \\
I_{tE} &= (\mathcal{B}R^{-1}\mathcal{B}^{\top} + Y_L)V_r + \bar{I}_L, \\
\bar{I}_E &= R^{-1}\mathcal{B}^{\top}V_r.
\end{align*}
Since $R^{-1} > 0$ and $Y_L > 0$ are diagonal positive definite matrices, and $\mathcal{B}R^{-1}\mathcal{B}^{\top}$ is positive semidefinite for any connected DC MG topology, the matrix $(\mathcal{B}R^{-1}\mathcal{B}^{\top} + Y_L)$ is positive definite and hence invertible. Therefore, $(V_E, I_{tE}, \bar{I}_E)$ exist and are unique for any given $V_r$ and $\bar{I}_L$. The control input $u_E$ is then uniquely determined by $u_E = [I + R_t(\mathcal{B}R^{-1}\mathcal{B}^{\top} + Y_L)]V_r + R_t\bar{I}_L$.
\end{proof}


\begin{remark}\label{Lm:currentsharing}
At the equilibrium, to acheive proportional current sharing among DG units we require:
\begin{equation}\label{Eq:currentsharing}
\frac{I_{tiE}}{P_{ni}} = I_s \Longleftrightarrow I_{tiE} = P_{ni}I_s, \ \forall i\in\mathbb{N}_N,
\end{equation}
where $P_{ni}$ is the power rating of $\Sigma_i^{DG}$ and $I_s$ represents the common current sharing ratio for the entire DC MG. This can be expressed in a vectorized form as:
\begin{equation}\label{Eq:vec_currentsharing}
I_{tE} = P_{n} \mathbf{1}_N I_s,
\end{equation}
where $\mathbf{1}_N=\begin{bmatrix}1, 1 , \ldots , 1\end{bmatrix}^\top\in\mathbb{R}^{N}$ and $P_n\triangleq\diag([P_{ni}]_{i\in\mathbb{N}_N})$. This relationship ensures that, at the equilibrium, the DG currents are distributed proportionally to their power ratings $P_n$. Therefore, in light of \eqref{Eq:vec_currentsharing}, the steady control input $u_{iS}$ is determined from \eqref{Eq:controlEquil}, in the form:
\begin{equation}\label{Eq:Rm:SteadyStateControl}
    u_{iS} = V_{ri} + R_{ti}P_{ni}I_s.
\end{equation}
\end{remark}

In conclusion, using Lm. \ref{Lm:equilibrium} and Rm. \ref{Lm:currentsharing}, for the equilibrium of DC MG, we require:
\begin{equation}
    \begin{aligned}
        V_E &= V_r, \\
        I_{tE} &= (\mathcal{B}R^{-1}\mathcal{B}^\T+Y_L)V_r + \bar{I}_L = P_n \textbf{1}_N I_s,\\
        \bar{I}_E &= R^{-1} \mathcal{B}^\T V_r.
    \end{aligned}
\end{equation}

These equations establish the existence of an equilibrium point.
The current sharing objective imposes additional constraints on the selection of the reference voltage vector $V_r$ and the sharing coefficient $I_s$. These design variables must satisfy the equilibrium constraint from Lm. \ref{Lm:equilibrium} while respecting DC MG constraints.

To achieve the desired equilibrium point, the steady-state control input $u_{iS}$ from equation \eqref{Eq:Rm:SteadyStateControl} must equal the equilibrium control input $u_{E}$, i.e., $u_{iS} = u_E$. This ensures that the proposed hierarchical control strategy can maintain the system at the desired operating point. Therefore, the overall equilibrium control input is:
\begin{equation}
    u_E = [\I + R_t(\mathcal{B}R^{-1}\mathcal{B}^\T + Y_L)] V_r + R_t\bar{I}_L = V_r + R_t I_{tE}.
\end{equation}

The current sharing objective imposes additional constraints on the selection of the reference voltage vector $V_r$ and the sharing coefficient $I_s$. These design variables must satisfy the equilibrium constraint from Lm. \ref{Lm:equilibrium} while respecting DC MG constraints.



\begin{theorem}\label{Th:VRegulation_CSharing}
The optimal reference voltage vector $V_r$ and current sharing coefficient $I_s$ can be determined by solving the constrained optimization problem:


\begin{equation}
\begin{aligned}
\mbox{Find: } \quad & \alpha_V\Vert V_r - \bar{V}_r\Vert^2 + \alpha_I I_s,\\
\mbox{Sub. to:} \quad & P_n \mathbf{1}_N I_s - (\mathcal{B}R^{-1}\mathcal{B}^\T + Y_L)V_r = \bar{I}_L, \\
& V_{\min} \leq V_r \leq V_{\max}, \\
& 0 \leq I_s \leq 1, \\
\end{aligned}
\end{equation}
where $V_{\min}$ and $V_{\max}$ represent the acceptable voltage bounds for the DC MG, $\bar{V}_r$ is the desired reference voltage vector, $\alpha_V > 0$ and $\alpha_I > 0$ are normalizing weights.
\end{theorem}
\begin{proof}
The proof follows directly from Lm. \ref{Lm:existence_uniqueness} and Rm. \ref{Lm:currentsharing}. The Lm. \ref{Lm:existence_uniqueness} establishes the existence and uniqueness of equilibrium variables for any given $V_r$ and $\bar{I}_L$. Rm. \ref{Lm:currentsharing} provides the current sharing constraint that must be satisfied at equilibrium. The optimization formulation combines these results: the equality constraint ensures compatibility between voltage regulation and current sharing objectives, while the bounds ensure practical operating limits. Since both the objective function and constraints are convex, the optimization problem has a unique solution when feasible.
\end{proof}

This formulation ensures proper system operation through multiple aspects. The equality constraint ensures that the equilibrium point is compatible with the current sharing objective across all DG units. The bounds on $V_r$ ensure system operating point to be within safe and efficient limits. Furthermore, the constraint on $I_s$ ensures the sharing coefficient remains within feasible bounds for practical implementation. The feasible solution set of this optimization problem provides valid combinations of $V_r$ and $I_s$ that establish equilibrium points with desired voltage regulation and current sharing characteristics while respecting system constraints. Convergence to these equilibrium points will be ensured by the subsequent controller design.




\subsection{Error Dynamics Networked System}\label{Sec:ErrorDynamics}
The network system representation described in Sec. \ref{Networked System Model} can be analyzed by considering the error dynamics around the identified equilibrium point in Lm. \ref{Lm:equilibrium}. As we will see in the sequel, the resulting error dynamics can be seen as a networked system (called the networked error system) comprised of DG error subsystems, line error subsystems, external disturbance inputs, and performance outputs.


We first define error variables that capture deviations from the identified equilibrium:
\begin{subequations}
\begin{align}
    \tilde{V}_i &\triangleq V_i - V_{iE} = V_i - V_{ri}, \label{Eq:errorvariable1}\\
    \tilde{I}_{ti} &\triangleq I_{ti} - I_{tiE} = {I_{ti} - P_{ni}I_s}, \label{Eq:errorvariable2}\\
     \tilde{v}_i &\triangleq v_i - v_{iE}, \label{Eq:errorvariable3}\\
    \tilde{I}_l &\triangleq I_{l} - \bar{I}_{lE} = I_l - \frac{1}{R_l}\sum_{i\in\E_l}\mathcal{B}_{il}V_{ri},
\end{align}
\end{subequations}
where $v_{iE}$ represents the equilibrium value of the integrator state $v_i$, which is a constant determined by the initial conditions and the equilibrium operating point. Now, considering the dynamics \eqref{line}, \eqref{Eq:ss:voltages}-\eqref{Eq:ss:ints}, equilibrium point established in Lm. \ref{Lm:equilibrium}, and the proposed a hierarchical control strategy $u_i(t)$ \eqref{controlinput}, the error dynamics can then be derived as follows. The voltage error dynamic can be achieved using \eqref{Eq:ss:voltages} and \eqref{Eq:errorvariable1}:

\begin{equation*}\label{Eq:currenterrordynamic}
\begin{aligned}
    \dot{\tilde{V}}_i = &-\frac{Y_{Li}}{C_{ti}}(\tilde{V}_i+V_{ri}) + \frac{1}{C_{ti}}(\tilde{I}_{ti} +P_{ni}I_s) +\frac{1}{C_{ti}}w_{vi}\\
    &- \frac{1}{C_{ti}}\bar{I}_{Li} - \frac{1}{C_{ti}}\sum_{l\in \mathcal{E}i} \mathcal{B}_{il}(\tilde{I}_l + \frac{1}{R_l}\sum_{j\in\E_l}\mathcal{B}_{jl}V_{rj}),\\
    \equiv& \frac{1}{C_{ti}}\Big(\phi_V +  \psi_V +w_{vi}\Big),
\end{aligned}
\end{equation*}
where 
\begin{subequations}
\begin{align}
   \Phi_V =& -Y_{Li}\tilde{V}_i + \tilde{I}_{ti} - \sum_{l\in \mathcal{E}i} \mathcal{B}_{il}\tilde{I}_l,\\
    \Psi_V =& -Y_{Li}V_{ri}+P_{ni}I_s - \bar{I}_{Li} - \sum_{l\in \mathcal{E}i} \frac{\mathcal{B}_{il}}{R_l}\sum_{j\in\E_l}\mathcal{B}_{jl}V_{rj}. \label{Eqvoltageerror}
\end{align}
\end{subequations}

The current error dynamic can be achieved b using \eqref{Eq:ss:currents} and \eqref{Eq:errorvariable2}:
\begin{equation*}\label{Eq:voltageerrordynamic}
\begin{aligned}
    \dot{\tilde{I}}_{ti} &= -\frac{1}{L_{ti}}(\tilde{V}_i+V_{ri}) - \frac{R_{ti}}{L_{ti}}(\tilde{I}_{ti} +P_{ni}I_s) + \frac{1}{L_{ti}}w_{ci} \\
    &+ \frac{1}{L_{ti}}(u_{iS}+k_{i0}^P\tilde{V}_i+k_{io}^I\tilde{v}_i+\sum_{j\in\bar{\mathcal{F}}_i^-} k_{ij}(\frac{\tilde{I}_{tj}}{P_{nj}} - \frac{\tilde{I}_{ti}}{P_{ni}})) ,\\
    \equiv& \frac{1}{L_{ti}}\Big(\Phi_I + \Psi_I+w_{ci}\Big),
\end{aligned}
\end{equation*}
where
\begin{subequations}
\begin{align}
    \Phi_I =&-\tilde{V}_i - R_{ti}\tilde{I}_{ti} + k_{io}^p\tilde{V}_i + k_{io}^I\tilde{v}_i
    +\sum_{j\in\bar{\mathcal{F}}_i^-} k_{ij}(\frac{\tilde{I}_{tj}}{P_{nj}} - \frac{\tilde{I}_{ti}}{P_{ni}})),\\
     \Psi_I =& -V_{ri}-R_{ti}P_{ni}I_s + u_{iS}. \label{Eqcurrenteerror}
\end{align}
\end{subequations}

The integral error dynamics can be achieved by using \eqref{Eq:ss:ints} and \eqref{Eq:errorvariable3}:
\begin{equation*}\label{Eq:integralerrordynamic}
    \dot{\tilde{v}}_i = \tilde{V}_i.
\end{equation*}



It is worth noting that, as a consequence of the equilibrium analysis and the steady state control input selection (see \eqref{Eq:currentsharing} and \eqref{Eq:Rm:SteadyStateControl}), the terms $\Psi_V=0$ \eqref{Eqvoltageerror} and $\Psi_I=0$ \eqref{Eqcurrenteerror}. Therefore, for each DG error subsystem $\tilde{\Sigma}_i^{DG}, i\in\N_N$, we have an error state vector $\tilde{x}_i = \begin{bmatrix} \tilde{V}_i, \tilde{I}_{ti}, \tilde{v}_i \end{bmatrix}^\T$ with the dynamics: 
\begin{equation}\label{Eq:DG_error_dynamic}
    \dot{\tilde{x}}_i = (A_i + B_i K_{i0})\tilde{x}_i  + \tilde{u}_i, 
\end{equation}
where $\tilde{u}_i$ represents the interconnection input combining the effects of both line currents and other DG states, defined as
\begin{equation}
    \tilde{u}_i \triangleq 
    \begin{bmatrix}
    \sum_{l\in\mathcal{E}_i}\Bar{C}_{il}\tilde{\Bar{x}}_l \\
       \sum_{j\in\bar{\mathcal{F}}_i^-}K_{ij}\tilde{x}_j\\
       0
    \end{bmatrix} + E_iw_i ,
\end{equation}
with system and disturbance matrices, $A_i$, $B_i$ and $E_i$ being the same as before in \eqref{Eq:DGCompact}.


Following similar steps, we can obtain the dynamics of the transmission line error subsystem $\tilde{\Sigma}_l^{Line}, l\in\N_L$ as:
\begin{equation}\label{Eq:Line_error_dynamic}
    \dot{\tilde{\bar{x}}}_l = \bar{A}_l\tilde{\bar{x}}_l + \tilde{\bar{u}}_l,  
\end{equation}
where $\tilde{\bar{x}}_l=\tilde{I}_l$ and $\tilde{\bar{u}}_l$ represents the line interconnection input influenced by other DG voltages and disturbances, defined as:
\begin{equation}
    \tilde{\bar{u}}_l  \triangleq \sum_{i\in\mathcal{E}_l}B_{il}\tilde{V}_i  + \bar{E}_l \bar{w}_l,
\end{equation}
where the disturbance matrix $\bar{E}_l$ is defined in \eqref{Eq:LineCompact}. To ensure robust stability (dissipativity) of this networked error system, we define performance outputs as follows. For each DG error subsystem $\tilde{\Sigma}_i^{DG}, i\in\N_N$, we define the performance output as:
\begin{equation}
z_i(t) = H_i\tilde{x}_i(t),
\end{equation}
where $H_i$ is selected as $H_i = \I$. 
Similarly, for each line error subsystem $\tilde{\Sigma}_l^{Line}, l\in\N_L$, we define the performance output as:
\begin{equation}
\bar{z}_l(t) = \bar{H}_l\tilde{\bar{x}}_l(t),
\end{equation}
where $\bar{H}_l$ is selected as $\bar{H}_l = \I$. 

Upon respectively vectorizing these performance outputs over all $i\in\N_N$ and $l\in\N_L$, we obtain: 
\begin{equation}\label{z}
\begin{aligned}
     z = H\tilde{x}\quad \mbox{ and } \quad 
     \bar{z} = \bar{H}\tilde{\bar{x}},
\end{aligned}
\end{equation} 
where $H \triangleq \diag([H_i]_{i\in\N_N})$ and $\bar{H} \triangleq \diag([\bar{H}_l]_{l\in\N_L})$.

To represent the networked error dynamics in a compressed manner, we consolidate the performance outputs and disturbance input vectors, respectively, as:
\begin{equation}
    \begin{aligned}
        z_c \triangleq \bm{z \\ \bar{z}}\quad \mbox{ and } \quad 
        w_c \triangleq \bm{w \\ \bar{w}}.
    \end{aligned}
\end{equation}
The consolidated disturbance vector $w_c$ affects the networked error dynamics, particularly the DG error subsystems and the line error subsystems, respectively, through the consolidated disturbance matrices $E_c$ and $\bar{E}_c$, defined as:
\begin{equation}
    \begin{aligned}
        E_c \triangleq \bm{E & \0}\quad \mbox{ and }\quad 
        \bar{E}_c \triangleq \bm{\0 & \bar{E}},
    \end{aligned}
\end{equation}
where recall that $E \triangleq \diag([E_i]_{i \in \N_N})$ and $\bar{E} \triangleq \diag(\bar{E}_l: l \in \N_L)$. The zero blocks in the $E_c$ and $\bar{E}_c$ indicate that line disturbances do not directly affect DG error subsystem inputs and vice versa. Analogously, the dependence of consolidated performance outputs on the networked error system states can be described using consolidated performance matrices 
\begin{equation}
    \begin{aligned}
        H_c \triangleq \bm{H \\ \0}\quad \mbox{ and }\quad 
        \bar{H}_c \triangleq \bm{\0 \\ \bar{H}}.
    \end{aligned}
\end{equation}

With these definitions and the derived error subsystem dynamics \eqref{Eq:DG_error_dynamic} and \eqref{Eq:Line_error_dynamic}, it is easy to see that the closed-loop error dynamics of the DC MG can be modeled as a networked error system as shown in Fig. \ref{Fig.DissNetError}. In there, the interconnection relationship between the error subsystems, disturbance inputs and performance outputs is described by:
\begin{equation}
    \begin{bmatrix}
        \tilde{u} \\ \tilde{\bar{u}} \\ z_c
    \end{bmatrix} = M \begin{bmatrix}
        \tilde{x} \\ \tilde{\bar{x}} \\ w_c
    \end{bmatrix},
\end{equation}
where the interconnection matrix $M$ takes the form:
\begin{equation}\label{Eq:NetErrSysMMat}
    M \triangleq \bm{M_{\tilde{u}x} & M_{\tilde{u}\bar{x}} & M_{\tilde{u}w_c}\\
    M_{\tilde{\bar{u}}x} &  M_{\tilde{\bar{u}}\bar{x}} &  M_{\tilde{\bar{u}}w_c}\\
    M_{z_cx} &  M_{z_c\bar{x}} &  M_{z_cw_c}
        } \equiv 
        \bm{K & \bar{C} & E_c\\
        C & \0 & \bar{E}_c\\
        H_{c} & \bar{H}_{c} & \0
        }.
\end{equation}

\begin{figure}
    \centering
    \includegraphics[width=0.9\columnwidth]{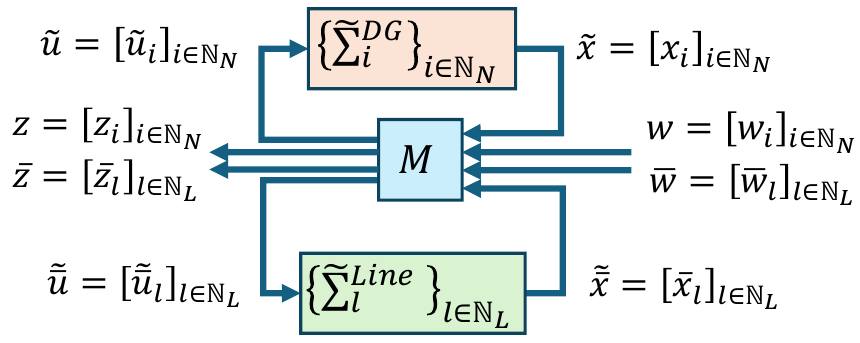}
    \caption{DC MG  error dynamics as a networked system with disturbance inputs and performance outputs.}
    \label{Fig.DissNetError}
\end{figure}

\section{Dissipativity-Based Control and Topology Co-Design}\label{Passivity-based Control}
In this section, we first analyze the dissipativity properties of the DG and line error subsystems. We then formulate the global control and topology co-design problem using the networked error dynamics from Sec. \ref{Sec:ErrorDynamics}. To support the feasibility of this global co-design, we identify necessary conditions on subsystem passivity indices and embed them into a local controller synthesis problem.

\subsection{Error Subsystem Dissipativity Properties}
Consider the DG error subsystem $\tilde{\Sigma}_i^{DG},i\in\mathbb{N}_N$ \eqref{Eq:DG_error_dynamic} to be $X_i$-EID with
\begin{equation}\label{Eq:XEID_DG}
    X_i=\begin{bmatrix}
        X_i^{11} & X_i^{12} \\ X_i^{21} & X_i^{22}
    \end{bmatrix}\triangleq
    \begin{bmatrix}
        -\nu_i\mathbf{I} & \frac{1}{2}\mathbf{I} \\ \frac{1}{2}\mathbf{I} & -\rho_i\mathbf{I}
    \end{bmatrix},
\end{equation}
where $\rho_i$ and $\nu_i$ are the passivity indices of $\tilde{\Sigma}_i^{DG}$. In other words, consider $\tilde{\Sigma}_i^{DG}, i\in\mathbb{N}_N$ as IF-OFP($\nu_i,\rho_i$).

Similarly, consider the line error subsystem $\tilde{\Sigma}_l^{Line},l\in\mathbb{N}_L$ \eqref{Eq:Line_error_dynamic} to be $\bar{X}_l$-EID with
\begin{equation}\label{Eq:XEID_Line}
    \bar{X}_l=\begin{bmatrix}
        \bar{X}_l^{11} & \bar{X}_l^{12} \\ \bar{X}_l^{21} & \bar{X}_l^{22}
    \end{bmatrix}\triangleq
    \begin{bmatrix}
        -\bar{\nu}_l\mathbf{I} & \frac{1}{2}\mathbf{I} \\ \frac{1}{2}\mathbf{I} & -\bar{\rho}_l\mathbf{I}
    \end{bmatrix},
\end{equation}
where $\bar{\rho}_l$ and $\bar{\nu}_l$ are the passivity indices of $\tilde{\Sigma}_l^{Line}$.  The following lemma characterizes these passivity indices.

\begin{lemma}\label{Lm:LineDissipativityStep}
For each line $\tilde{\Sigma}_l^{Line}, l\in\N_L$ \eqref{Eq:LineCompact}, its passivity indices $\bar{\nu}_l$,\,$\bar{\rho}_l$ assumed in (\ref{Eq:XEID_Line}) are such that the LMI problem: 
\begin{equation}\label{Eq:Lm:LineDissipativityStep1}
\begin{aligned}
\mbox{Find: }\ &\bar{P}_l, \bar{\nu}_l, \bar{\rho}_l\\
\mbox{Sub. to:}\ &\bar{P}_l > 0, \ 
\begin{bmatrix}
    \frac{2\bar{P}_lR_l}{L_l}-\bar{\rho}_l & -\frac{\bar{P}_l}{L_l}+\frac{1}{2}\\
    \star & -\bar{\nu}_l
\end{bmatrix}
\normalsize
\geq0, 
\end{aligned}
\end{equation}
is feasible. The maximum feasible values for $\bar{\nu}_l$ and $\bar{\rho}_l$ are $\bar{\nu}_l^{\max}=0$ and $\bar{\rho}_l^{\max}=R_l$ respectively, when $\bar{P}_l =  \frac{L_l}{2}$. 
\end{lemma}

\begin{proof}
For each $\tilde{\Sigma}^{Line}_l$, $l \in \mathcal{N}_L$ from \eqref{Eq:Line_error_dynamic}, we must verify that it satisfies $\bar{X}_l$-EID with the passivity indices given by \eqref{Eq:XEID_Line}. For this, we can apply Prop. \ref{Prop:linear_X-EID} with the given system matrices from \eqref{Eq:LineMatrices} and the specified dissipativity supply rate in \eqref{Eq:XEID_Line}, leading to the LMI condition:
\begin{equation*}
\begin{bmatrix} 2\bar{P}_l\frac{R_l}{L_l} - \bar{\rho}_l & -\bar{P}_l\frac{1}{L_l} + \frac{1}{2} \\ \star & -\bar{\nu}_l \end{bmatrix} \geq 0.
\end{equation*}

Using Lm. \ref{Lm:Schur_comp}, this is positive semidefinite if and only if:
\begin{equation*} 
\begin{aligned} 
\bar{\nu}_l \leq 0 \quad \text{and} \quad  2\bar{P}_l\frac{R_l}{L_l} - \bar{\rho}_l - \frac{(-\bar{P}_l\frac{1}{L_l} + \frac{1}{2})^2}{-\bar{\nu}_l} \geq 0.
\end{aligned}
\end{equation*}

To obtain maximum passivity indices, we choose $\bar{\nu}_l = 0$, which requires:
\begin{equation*}
   \begin{aligned}
2\bar{P}_l\frac{R_l}{L_l} - \bar{\rho}_l \geq 0 \quad \text{and} \quad
-\bar{P}_l\frac{1}{L_l} + \frac{1}{2} = 0, 
   \end{aligned} 
\end{equation*}
which gives $\bar{P}_l = \frac{L_l}{2}$. Substituting $\bar{P}_l = \frac{L_l}{2}$ into the first constraint yields $R_l - \bar{\rho}_l \geq 0$, implying $\bar{\rho}_l \leq R_l$. Therefore, the maximum feasible values are $\bar{\nu}_l^{\max} = 0$ and $\bar{\rho}_l^{\max} = R_l$, when $\bar{P}_l = \frac{L_l}{2}$.
\end{proof}

Although we could identify the conditions required for the passivity indices of the line error dynamics \eqref{Eq:Line_error_dynamic}, achieving a similar feat for the DG error dynamics is not straightforward due to the involved CPL nonlinearities (see \eqref{Eq:DG_error_dynamic}). This challenge is addressed in the following subsection.

\subsection{Global Control and Topology Co-Design}

The interconnection matrix $M$ (\ref{Eq:MMatrix}), particularly its block $M_{\tilde{u}x}=K$, can be synthesized by applying our subsystem EID properties to Prop. \ref{synthesizeM}. By synthesizing $K=[K_{ij}]_{i,j\in\N_N}$, we can uniquely compute the distributed global controller gains $\{k_{ij}^I:i,j\in\mathbb{N}_N\}$ (\ref{k_ij}) and the required communication topology $\mathcal{G}^c$.
The following theorem formulates this distributed global controller and communication topology co-design problem.

\begin{theorem}\label{Th:CentralizedTopologyDesign}
The closed-loop dynamics of the DC MG are illustrated in Fig. \ref{Fig.DissNetError}, can be made finite-gain $L_2$-stable with an $L_2$-gain $\gamma$ (where $\Tilde{\gamma}\triangleq \gamma^2<\bar{\gamma}$ and $\bar{\gamma}$ is prespecified) from unknown disturbances $w_c(t)$ to performance output $z_c(t)$, by synthesizing the interconnection matrix block $M_{\tilde{u}x}=K$ (\ref{Eq:MMatrix}) via solving the LMI problem:
\begin{equation}
\label{Eq:Th:CentralizedTopologyDesign0}
\begin{aligned}
&\min_{\substack{Q,\{p_i: i\in\N_N\},\\
\{\bar{p}_l: l\in\N_L\}, \tilde{\gamma}}} &&\sum_{i,j\in\N_N} c_{ij} \Vert Q_{ij} \Vert_1 +c_1 \tilde{\gamma} + \alpha\text{tr}(s_W),\\
&
\mbox{ Sub. to:}  
&&p_i > 0,\ \forall i\in\N_N,\ 
\bar{p}_l > 0,\ \forall l\in\N_L,\\   
& &&0 < \tilde{\gamma} < \bar{\gamma},  \\
& && W + s_W > 0, \ s_W \geq 0, \\
& &&\text{tr}(s_W) \leq \eta
\mbox{ and \eqref{globalcontrollertheorem}},\\
& && Q_I P_n \textbf{1}_N = 0,
\end{aligned}
\end{equation}
as $K = (\textbf{X}_p^{11})^{-1} Q$ and $Q_I = \bm{Q_{ij}^{2,2}}_{i,j\in\N_N}$, where 
$\textbf{X}^{12} \triangleq 
\diag([-\frac{1}{2\nu_i}\I]_{i\in\N_N})$, 
$\textbf{X}^{21} \triangleq (\textbf{X}^{12})^\T$,
$\Bar{\textbf{X}}^{12} \triangleq 
\diag([-\frac{1}{2\Bar{\nu}_l}\I]_{l\in\N_L})$,
$\Bar{\textbf{X}}^{21} \triangleq (\Bar{\textbf{X}}^{12})^\T$, 
$\textbf{X}_p^{11} \triangleq 
\diag([-p_i\nu_i\I]_{i\in\N_N})$, 
$\textbf{X}_p^{22} \triangleq 
\diag([-p_i\rho_i\I]_{i\in\N_N})$, 
$\Bar{\textbf{X}}_{\bar{p}}^{11} 
\triangleq \diag([-\bar{p}_l\bar{\nu}_l\I]_{l\in\N_L})$, 
$\Bar{\textbf{X}}_{\bar{p}}^{22} 
\triangleq \diag([-\bar{p}_l\bar{\rho}_l\I]_{l\in\N_L})$ and $\tilde{\Gamma} \triangleq \tilde{\gamma}\I$. 
The structure of $Q\triangleq[Q_{ij}]_{i,j\in\N_N}$ mirrors that of $K\triangleq[K_{ij}]_{i,j\in\N_N}$ (i.e., the first and third rows are zeros in each block $Q_{ij}$, see \eqref{k_ij}). 
The coefficients $c_1>0$ and $c_{ij}>0,\forall i,j\in\N_N$ are the cost coefficients corresponding to the $L_2$-gain from unknown disturbances and communication links, respectively. Note that $W$ represents the LMI matrix from \eqref{globalcontrollertheorem}, $s_W$ is a symmetric slack matrix, $\alpha > 0$ is the slack penalty weight, and $\eta > 0$ is a bound on the total slack magnitude. Here, $\bar{\gamma}$ is the prespecified positive constant that represents the upper bounds on the $L_2$-gain for disturbance attenuation.
\end{theorem}

\begin{proof}
The proof follows by treating the closed-loop DC MG (shown in Fig. \ref{Fig.DissNetError}) as a networked error system and applying the subsystem dissipativity properties from \eqref{Eq:XEID_DG} and \eqref{Eq:XEID_Line} to the interconnection topology synthesis framework of Prop. \ref{synthesizeM}. We modeled the DG error subsystems as IF-OFP($\nu_i,\rho_i$) and line error subsystems as  IF-OFP($\Bar{\nu}_l,\Bar{\rho}_l$). The LMI problem \eqref{Eq:Th:CentralizedTopologyDesign0} is formulated to ensure that the networked error system is \textbf{Y}-EID, thereby ensuring finite-gain $L_2$-stability with gain $\gamma$ from disturbances $w_c$ to performance outputs $z_c$. The objective function in \eqref{Eq:Th:CentralizedTopologyDesign0} consists of three terms: communication cost ($\sum_{i,j\in\N_N} c_{ij} \Vert Q_{ij} \Vert_1$), control cost $c_1 \tilde{\gamma}$, and a numerical stability term ($\alpha\text{tr}(s_W)$). Minimizing this function while satisfying LMI constraints simultaneously optimizes the communication topology (by synthesizing $K = (\textbf{X}_p^{11})^{-1} Q$) and robust stability (by minimizing $\tilde{\gamma}$) while ensuring the given specification $\gamma^2<\bar{\gamma}$. Note that the proposed cost function in \eqref{Eq:Th:CentralizedTopologyDesign0} jointly optimizes the communication topology (while inducing sparsity \cite{WelikalaJ22022}) and robust stability (i.e., $L_2$-gain $\gamma$) of the DC MG. The resulting controller and topology achieve voltage regulation and current sharing in the presence of disturbances. 
\end{proof}

The slack matrix $s_W$ improves the numerical conditioning of the LMI constraints by providing additional degrees of freedom in optimization. The penalty term $\alpha \text{tr}(s_W)$ prevents the excessive use of slack variables while preserving the feasibility of the problem. The bound $\text{tr}(s_W) \leq \eta$ limits how much the constraints can be relaxed, which ensures that the solution remains feasible. This approach achieves a balance between computational tractability and robust performance guarantees for the DC MG.



\begin{figure*}[!hb]
\vspace{-5mm}
\centering
\hrulefill
\begin{equation}\label{globalcontrollertheorem}
\scriptsize
	W = \bm{
		\textbf{X}_p^{11} & \0 & \0 & Q & \textbf{X}_p^{11}\Bar{C} &  \textbf{X}_p^{11}E_c \\
		\0 & \bar{\textbf{X}}_{\bar{p}}^{11} & \0 & \Bar{\textbf{X}}_{\Bar{p}}^{11}C & \0 & \bar{\textbf{X}}_{\bar{p}}^{11}\bar{E}_c \\
		\0 & \0 & \I & H_c & \bar{H}_c & \0 \\
		Q^\T & C^\T\Bar{\textbf{X}}_{\Bar{p}}^{11} & H_c^\T & -Q^\T\textbf{X}^{12}-\textbf{X}^{21}Q-\textbf{X}_p^{22} & -\textbf{X}^{21}\textbf{X}_{p}^{11}\bar{C}-C^\T\bar{\textbf{X}}_{\bar{p}}^{11}\bar{\textbf{X}}^{12} & -\textbf{X}^{21}\textbf{X}_p^{11}E_c \\
		\Bar{C}^\T\textbf{X}_p^{11} & \0 & \bar{H}_c^\T & -\Bar{C}^\T\textbf{X}_p^{11}\textbf{X}^{12}-\bar{\textbf{X}}^{21}\Bar{\textbf{X}}_{\Bar{p}}^{11}C & -\bar{\textbf{X}}_{\bar{p}}^{22} & -\bar{\textbf{X}}^{21}\Bar{\textbf{X}}_{\Bar{p}}^{11}\bar{E}_c \\ 
		E_c^\T\textbf{X}_p^{11} & \bar{E}_c^\T\Bar{\textbf{X}}_{\Bar{p}}^{11} & \0 & -E_c^\T\textbf{X}_p^{11}\textbf{X}^{12} & -\bar{E}_c^\T\Bar{\textbf{X}}_{\Bar{p}}^{11}\bar{\textbf{X}}^{12} & \tilde{\Gamma} \\
	}\normalsize 
 >0 
\end{equation}
\end{figure*}

\begin{figure*}[!hb]
\vspace{-5mm}
\centering
\begin{equation}\label{Eq:Neccessary_condition}
\scriptsize
	\bm{
		-p_i\nu_i & 0 & 0 & 0 & -p_i\nu_i\bar{C}_{il} & -p_i\nu_i\\
		0 & -\bar{p}_l\bar{\nu}_l & 0 & -\bar{p}_l\bar{\nu}_lC_{il} & 0 & -\bar{p}_l\bar{\nu}_l \\
		0 & 0 & 1 & 1 & 1 & 0 \\
		0 & -C_{il}\bar{\nu}_l\bar{p}_l & 1 & p_i\rho_i & -\frac{1}{2}p_i\bar{C}_{il}-\frac{1}{2}C_{il}\bar{p}_l & -\frac{1}{2}p_i \\
		-\bar{C}_{il}\nu_ip_i & 0 & 1 & -\frac{1}{2}\bar{C}_{il}p_i-\frac{1}{2}\bar{p}_lC_{il} & \bar{p}_l\bar{\rho}_l & -\frac{1}{2}\bar{p}_l \\ 
		-\nu_ip_i & -\bar{p}_l\bar{\nu}_l & 0 & -\frac{1}{2}p_i & -\frac{1}{2}\bar{p}_l & \tilde{\gamma}_i \\
	}\normalsize 
 >0,\ \forall l\in \mathcal{E}_i, \forall i\in\N_N
\end{equation}
\end{figure*}

\subsection{Necessary Conditions on Subsystem Passivity Indices}
The feasibility and effectiveness of the global co-design problem \eqref{Eq:Th:CentralizedTopologyDesign0} directly depend on the passivity indices $\{\nu_i,\rho_i:i\in\mathbb{N}_N\}$ \eqref{Eq:XEID_DG} and $\{\bar{\nu}_l,\bar{\rho}_l:l\in\mathbb{N}_L\}$ \eqref{Eq:XEID_Line} chosen for the DG \eqref{closedloopdynamic} and line \eqref{Eq:LineCompact} subsystems, respectively, as these parameters appear in the LMI constraint matrices. 

The local controller design using Co. \ref{Col.LTI_LocalController_XEID} determines specific passivity indices for the DGs \eqref{closedloopdynamic}, while the passivity analysis in  Lm. \ref{Lm:LineDissipativityStep} established passivity indices for the transmission lines \eqref{Eq:LineCompact}. Since these subsystem-level designs directly influence the global co-design problem, poor choices of local controllers or passivity indices can render the overall co-design infeasible and/or ineffective. 

Therefore, the local controller design and passivity analysis must account for conditions that ensure feasibility and effectiveness of the global co-design. The following lemma identifies these necessary conditions based on the global optimization problem \eqref{Eq:Th:CentralizedTopologyDesign0} in Th. \ref{Th:CentralizedTopologyDesign}. 

\begin{lemma}\label{Lm:CodesignConditions}
For the LMI conditions \eqref{Eq:Th:CentralizedTopologyDesign0} in Th. \ref{Th:CentralizedTopologyDesign} to hold, it is necessary that the DG and line passivity indices 
$\{\nu_i,\rho_i:i\in\mathbb{N}_N\}$ \eqref{Eq:XEID_DG} and  $\{\bar{\nu}_l,\bar{\rho}_l:l\in\mathbb{N}_L\}$ \eqref{Eq:XEID_Line} are such that the LMI problem: 
\begin{equation}\label{Eq:Lm:CodesignConditions}
\begin{aligned}
&\mbox{Find: }\ \ &&\{(\nu_i,\rho_i,\tilde{\gamma}_i):i\in\N_N\},\{(\Bar{\nu}_l,\bar{\rho}_l):l\in\N_L\},\\
&\mbox{Sub. to: }\
&&p_i > 0,\ \forall i\in\N_N, \  \bar{p}_l>0,\ \forall l\in\N_L,\\
& &&\mbox{and } \ \eqref{Eq:Neccessary_condition},
\end{aligned}
\end{equation} 
is feasible.
\end{lemma}

\begin{proof}
For the global co-design problem \eqref{Eq:Th:CentralizedTopologyDesign0} to be feasible, the matrix $W$ from \eqref{globalcontrollertheorem} must satisfy $W > 0$. Let $W = [W_{rs}]_{r,s\in\N_6}$ where each block $W_{rs}$ can be a block matrix of block dimensions $(N\times N)$, $(N\times L)$ or $(L \times L)$ depending on its location in $W$ (e.g., see blocks $W_{11}, W_{15}$ and $W_{22}$, respectively). Without loss of generality, let us denote $W_{rs} \triangleq [W_{rs}^{jm}]_{j\in\bar{J}(r),m\in\bar{M}(s)}$ where $\bar{J}(r),\bar{M}(s) \in \{N,L\}$. Inspired by \cite[Lm. 1]{WelikalaJ22022}, we can obtain an equivalent condition for $W>0$ as $\bar{W} \triangleq \text{BEW}(W) > 0$ where $\text{BEW}(W)$ is the ``block-elementwise'' form of $W$, created by combining appropriate inner-block elements of each of the blocks $W_{rs}$ to create a $6\times 6$ block-block matrix. 
Simply, $\bar{W} = [[W_{rs}^{j,m}]_{r,s \in \N_6}]_{j \in \N_{\bar{J}},m\in\N_{\bar{M}}}$. Considering only the diagonal blocks in  $\bar{W}$ and the implication $\bar{W} > 0 \implies [[W_{rs}^{j,m}]_{r,s \in \N_6}]_{j \in \N_{\bar{J}(r)},m\in\N_{\bar{M}(r)}} > 0 \iff$\eqref{Eq:Lm:CodesignConditions} (also recall the notations $C_{il} \triangleq -C_{ti}\bar{C}_{il}^\T$, $\Bar{C}_{il} \triangleq -C_{ti}^{-1}$). Therefore,  \eqref{globalcontrollertheorem} $\implies$ \eqref{Eq:Lm:CodesignConditions}, in other words, \eqref{Eq:Lm:CodesignConditions} is a set of necessary conditions for the feasibility of the global co-design constraint \eqref{globalcontrollertheorem}. Therefore, the necessary conditions in \eqref{Eq:Lm:CodesignConditions} are derived from the LMI constraints of the global co-design problem \eqref{Eq:Th:CentralizedTopologyDesign0}.
\end{proof}

Besides merely supporting the feasibility of the global co-design \eqref{globalcontrollertheorem}, the LMI problem \eqref{Eq:Lm:CodesignConditions}, through its inclusion of the constraint $0 \leq \tilde{\gamma}_i \leq \bar{\gamma}$ (which can also be embedded in the objective function), aims to improve the effectiveness (performance) of the global co-design \eqref{globalcontrollertheorem}.

\subsection{Local Controller Synthesis}
To enforce the identified necessary LMI conditions in Lm. \ref{Lm:CodesignConditions} on DG and line passivity indices, a local controller synthesis problem is formulated. This ensures numerical feasibility while guaranteeing system stability.

\begin{theorem}\label{Th:LocalControllerDesign}
Under the predefined DG parameters \eqref{Eq:DGCompact}, transmission line parameters \eqref{Eq:LineCompact} and design parameters $\{p_i: i\in\N_N\}$, $\{\bar{p}_l:l\in\N_L\}$, the necessary conditions in  \eqref{Eq:Th:CentralizedTopologyDesign0} hold if the local controller gains $\{K_{i0}, i\in\N_N\}$ (\ref{Controller}) and 
DG and line passivity indices 
$\{\nu_i,\rho_i:i\in\mathbb{N}_N\}$ \eqref{Eq:XEID_DG} and  $\{\bar{\nu}_l,\bar{\rho}_l:l\in\mathbb{N}_L\}$ \eqref{Eq:XEID_Line} are determined by solving the LMI problem:
\begin{equation}
\begin{aligned}\nonumber
&\mbox{Find: }\ \{(\tilde{K}_{i0}, P_i, \nu_i, \tilde{\rho}_i, \tilde{\gamma}_i):i\in\mathbb{N}_N\}, \{(\bar{P}_l, \bar{\nu}_l,\bar{\rho}_l):l\in\mathbb{N}_L\}, \\
&\mbox{Sub. to: }\  \\
&\
P_i > 0,\ 
\scriptsize
\bm{\tilde{\rho}_i\I & P_i & \0 \\
P_i &-\mathcal{H}(A_iP_i + B_i\Tilde{K}_{i0})& -\I + \frac{1}{2}P_i\\
\0 & -\I + \frac{1}{2}P_i & -\nu_i\I}
\normalsize
> 0,\ \forall i\in\mathbb{N}_N, \label{GetPassivityIndicesDGs}\\
&\ 
\bar{P}_l > 0,\ 
\scriptsize
\bm{
\frac{2\bar{P}_lR_l}{L_l}-\bar{\rho}_l & -\frac{\bar{P}_l}{L_l}+\frac{1}{2}\\
    \star & -\bar{\nu}_l
} 
\geq0,\ \forall l\in\mathbb{N}_L, \\
&\mbox{and } \ \eqref{Eq:Transformed_Necessary_condition},
\end{aligned}
\end{equation}
where $K_{i0} \triangleq \tilde{K}_{i0}P_i^{-1}$ and $\rho_i\triangleq\frac{1}{\tilde{\rho}_i}$. 
\end{theorem}

\begin{proof}
The proof proceeds as follows: (i) We start by considering the dynamic models of $\tilde{\Sigma}_i^{DG}$ and $\tilde{\Sigma}_i^{Line}$ as described in \eqref{Eq:DG_error_dynamic} and \eqref{Eq:Line_error_dynamic}, respectively. We then apply the LMI-based controller synthesis and analysis techniques from Co. \ref{Col.LTI_LocalController_XEID} and Lm. \ref{Lm:LineDissipativityStep} to enforce and identify the passivity indices of subsystems assumed in \eqref{Eq:XEID_DG} and \eqref{Eq:XEID_Line}, respectively. (ii) Next, we apply the LMI formulations from Co. \ref{Col.LTI_LocalController_XEID} to obtain the local controller gains $K_{i0}$ and the passivity indices $(\nu_i,\rho_i)$ for $\tilde{\Sigma}_i^{DG}$. Similarly, we apply Lm. \ref{Lm:LineDissipativityStep} to identify the passivity indices $(\bar{\nu}_l,\bar{\rho}_l)$ for $\tilde{\Sigma}_i^{Line}$. (iii) To handle the transformation from $\rho_i$ to $\tilde{\rho}_i$ where $\tilde{\rho}_i = \rho_i^{-1}$, we apply a congruence transformation to equation \eqref{Eq:Neccessary_condition}. Using the transformation matrix $T = \diag(1, 1, 1, 1/\rho_i, 1, 1)$, the term $p_i\rho_i$ at position (4,4)  of \eqref{Eq:Neccessary_condition} becomes $p_i\tilde{\rho}_i$. This transformation also affects other elements in position (4,4), resulting in the transformed equation \eqref{Eq:Transformed_Necessary_condition}. The bilinear terms involving $\bar{\nu}_l\tilde{\rho}_i$ in positions (2,4), (4,2), and other locations in position (4,4) require special handling. For each bilinear term $\bar{\nu}_l\tilde{\rho}_i$, we introduce an auxiliary variable $\xi_{il}$ and add the Schur complement constraint:
\begin{equation}
\begin{bmatrix} 
1 & \bar{\nu}_l & \tilde{\rho}_i \\
\bar{\nu}_l & s_1 & \xi_{il} \\
\tilde{\rho}_i & \xi_{il} & s_2
\end{bmatrix} \geq 0,
\end{equation}
where $s_1$ and $s_2$ are semidefinite variables determined during optimization. This constraint enforces $\xi_{il} \geq \bar{\nu}_l\tilde{\rho}_i$, allowing us to replace each bilinear term with $\xi_{il}$ in the transformed equation. (iv) Finally, we impose the necessary conditions on the passivity indices of subsystems identified in Lm. \ref{Lm:CodesignConditions} to support the feasibility and effectiveness of the global control and communication topology co-design approach presented in Th. \ref{Th:CentralizedTopologyDesign}. The constraint \eqref{Eq:Transformed_Necessary_condition} is derived from the necessary condition established in Lm. \ref{Lm:CodesignConditions}, ensuring that the local controllers and their passivity indices are compatible with the global co-design problem. By formulating this unified LMI problem, we obtain a one-shot approach to simultaneously design local controllers and determine passivity indices that guarantee the feasibility of the global co-design problem.
\end{proof}

\begin{figure*}[!hb]
\vspace{-5mm}
\centering
\hrulefill
\begin{equation}\label{Eq:Transformed_Necessary_condition}
\begin{bmatrix} 
1 & \bar{\nu}_l & \tilde{\rho}_i \\
\bar{\nu}_l & s_1 & \xi_{il} \\
\tilde{\rho}_i & \xi_{il} & s_2
\end{bmatrix} \geq 0,\
\scriptsize
	\bm{
		-p_i\nu_i & 0 & 0 & 0 & -p_i\nu_i\bar{C}_{il} & -p_i\nu_i\\
		0 & -\bar{p}_l\bar{\nu}_l & 0 & -\bar{p}_l\xi_{il}C_{il} & 0 & -\bar{p}_l\bar{\nu}_l \\
		0 & 0 & 1 & \tilde{\rho}_i & 1 & 0 \\
		0 & -C_{il}\xi_{il}\bar{p}_l & \tilde{\rho}_i & p_i\tilde{\rho}_i & -\frac{1}{2}p_i\bar{C}_{il}\tilde{\rho}_i-\frac{1}{2}C_{il}\bar{p}_l\tilde{\rho}_i & -\frac{1}{2}p_i\tilde{\rho}_i \\
		-\bar{C}_{il}\nu_ip_i & 0 & 1 & -\frac{1}{2}\bar{C}_{il}p_i\tilde{\rho}_i-\frac{1}{2}\bar{p}_lC_{il}\tilde{\rho}_i & \bar{p}_l\bar{\rho}_l & -\frac{1}{2}\bar{p}_l \\ 
		-\nu_ip_i & -\bar{p}_l\bar{\nu}_l & 0 & -\frac{1}{2}p_i\tilde{\rho}_i & -\frac{1}{2}\bar{p}_l & \tilde{\gamma}_i \\
	}\normalsize 
 >0,\ \forall l\in \mathcal{E}_i, \forall i\in\N_N
\end{equation}
\end{figure*}

\subsection{Overview}

In the proposed co-design process, we first select the design parameters $p_i, \forall i \in \mathbb{N}_N$ and $\bar{p}_l, \forall l \in \mathbb{N}_L$ (for more details, see \cite{Najafirad2024Ax1}). Next, we synthesize the local controllers using Th. \ref{Th:LocalControllerDesign}, to obtain DG and line passivity indices $\{\rho_i, \nu_i : \forall i \in \mathbb{N}_N\}$ \eqref{Eq:XEID_DG} and $\{\bar{\rho}_l, \bar{\nu}_l : \forall l \in \mathbb{N}_L\}$ \eqref{Eq:XEID_Line}, respectively. Finally, we synthesize the distributed global controller and the communication topology of the DC MG using Th. \ref{Th:CentralizedTopologyDesign}.

\section{Simulation Results}\label{Simulation}
In this section, we have evaluated the effectiveness of the proposed dissipativity-based controller in an islanded DC MG using SimPowerSystem/Simulink MATLAB environment. The DC MG test system consists of 4 DGs and 4 transmission lines, as illustrated in Fig. \ref{fig.physicalcommunicationtopology}, where each DG is equipped with its corresponding load consisting of local constant impedance $Y_L$ and constant current $\bar{I}_L$ loads. DGs are modeled as DC-DC converters with a nominal voltage of 120 V, and the reference voltage amplitude was set to $V_r=48$ V. Table I lists the mean values of the DG and line parameters. To create a heterogeneous testing environment that better reflects real-world conditions, we introduced random variations $\pm20\%$ from these mean values for each component. Simply, each DG, load and transmission line parameter was selected by adding a random variation to the corresponding mean nominal parameter value given in Table I. The physical topology of the DC MG was randomly generated using a random geometric graph generation technique with a connectivity parameter set to 0.6, following the approach in \cite{zhang2018robustness}. Fig. \ref{fig.physicalcommunicationtopology} shows the physical topology of the islanded DC MG used in our simulations.

\begin{figure}
    \centering
\includegraphics[width=0.8\columnwidth]{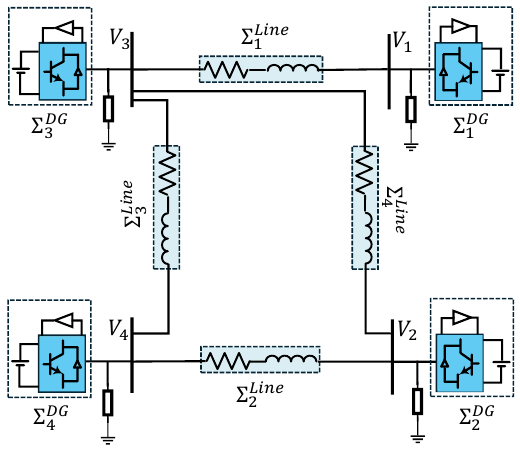}
    \caption{The physical topology of DC MG with 4 DGs and 4 lines.}
    \label{fig.physicalcommunicationtopology}
\end{figure}

Our implementation is carried out according to the systematic co-design process. First, we addressed the equilibrium design problem formulated in Th. \ref{Th:VRegulation_CSharing}, selecting balanced weighting parameters $\alpha_V = 1$ and $\alpha_I = 0.1$ to achieve an optimal trade-off between voltage regulation and current sharing objectives. This optimization yielded the reference voltage $V_r$ and the corresponding current sharing index $I_s$, which were subsequently used to calculate the steady-state control inputs according to \eqref{Eq:Rm:SteadyStateControl}. For the synthesis of the controller, we selected the design parameters $p_i = 0.1,\forall i\in\N_N$, and $\bar{p}_l = 0.01,\forall l\in\N_L$ to obtain the feasible solution. Using Th. \ref{Th:LocalControllerDesign}, we synthesized the local controllers, which yielded the DG and line passivity indices. These indices were then employed in the distributed global controller and communication topology co-design problem formulated in Th. \ref{Th:CentralizedTopologyDesign}. All LMI problems were efficiently solved using the MOSEK solver with default tolerance settings, demonstrating the computational practicality of our approach. 
\begin{table}
\caption{Parameters of the islanded DC MG test system\label{parameters}}
\centering
\renewcommand{\arraystretch}{1}

\makeatletter
\def\thickhline{%
  \noalign{\ifnum0=`}\fi\hrule \@height \thickarrayrulewidth \futurelet
   \reserved@a\@xthickhline}
\def\@xthickhline{\ifx\reserved@a\thickhline
               \vskip\doublerulesep
               \vskip-\thickarrayrulewidth
             \fi
      \ifnum0=`{\fi}}
\makeatother

\newlength{\thickarrayrulewidth}
\setlength{\thickarrayrulewidth}{2.5\arrayrulewidth}

\begin{tabular}{c | c c c c}
  \thickhline
  Characteristics & DG 1 & DG 2 & DG 3 & DG 4 \\
   \hline \hline
  Power rating & 800 W & 700 W & 700 W & 900 W \\
  \hline
  \multicolumn{1}{c|}{Internal resistance} & \multicolumn{4}{|c}{$R_t$ = 50 m$\Omega$} \\
  \hline
  \multicolumn{1}{c|}{Internal inductance} & \multicolumn{4}{|c}{$L_t$ = 10 mH} \\
  \hline
  \multicolumn{1}{c|}{Filter capacitance} & \multicolumn{4}{|c}{$C_t$ = 2.2 mF} \\
  \hline
  \multicolumn{1}{c|}{Constant impedance load} & \multicolumn{4}{|c}{$R_L = 1/Y_L$ = 5 $\Omega$} \\
  \hline
  \multicolumn{1}{c|}{Constant current load} & \multicolumn{4}{|c}{$\bar{I}_{L}$ = 3 A} \\
  \hline
  \multicolumn{1}{c|}{Line resistance} & \multicolumn{4}{|c}{$R_l$ = 20 m$\Omega$} \\
  \hline
  \multicolumn{1}{c|}{Line inductance} & \multicolumn{4}{|c}{$L_l$ = 10 mH} \\
  \hline
  \multicolumn{1}{c|}{Reference voltage} & \multicolumn{4}{|c}{ $V_{r}$ = 48 V} \\
  \hline
  \thickhline
\end{tabular}
\end{table}

\subsection{Control Layer Activation}
In this section, we analyze the hierarchical activation of different control layers and demonstrate their effectiveness in achieving voltage regulation and current sharing in DC MGs. Figures \ref{fig.outputvoltage} and \ref{fig.outputperunitcurrent} show the voltages and per-unit currents for each DG, respectively, as the control layers are progressively activated.
During the initial phase $t\in[0,1)$ s, no control input is applied to the PWM modules of the DC-DC converters, resulting in unregulated output voltages. Next, as shown in Fig. \ref{fig.outputvoltage}, steady-state control is activated at $t=1$ s, which causes a voltage deviation from the reference value and cannot fully address voltage regulation. Therefore, the local control layer is introduced at $t=3$ s to compensate for this deviation to achieve precise voltage regulation. The implementation of local controllers successfully restores the voltages to the desired reference voltage $V_r$, demonstrating the effectiveness of the local control layer for voltage regulation. 

\begin{figure}
    \centering
    \includegraphics[width=0.9\columnwidth]{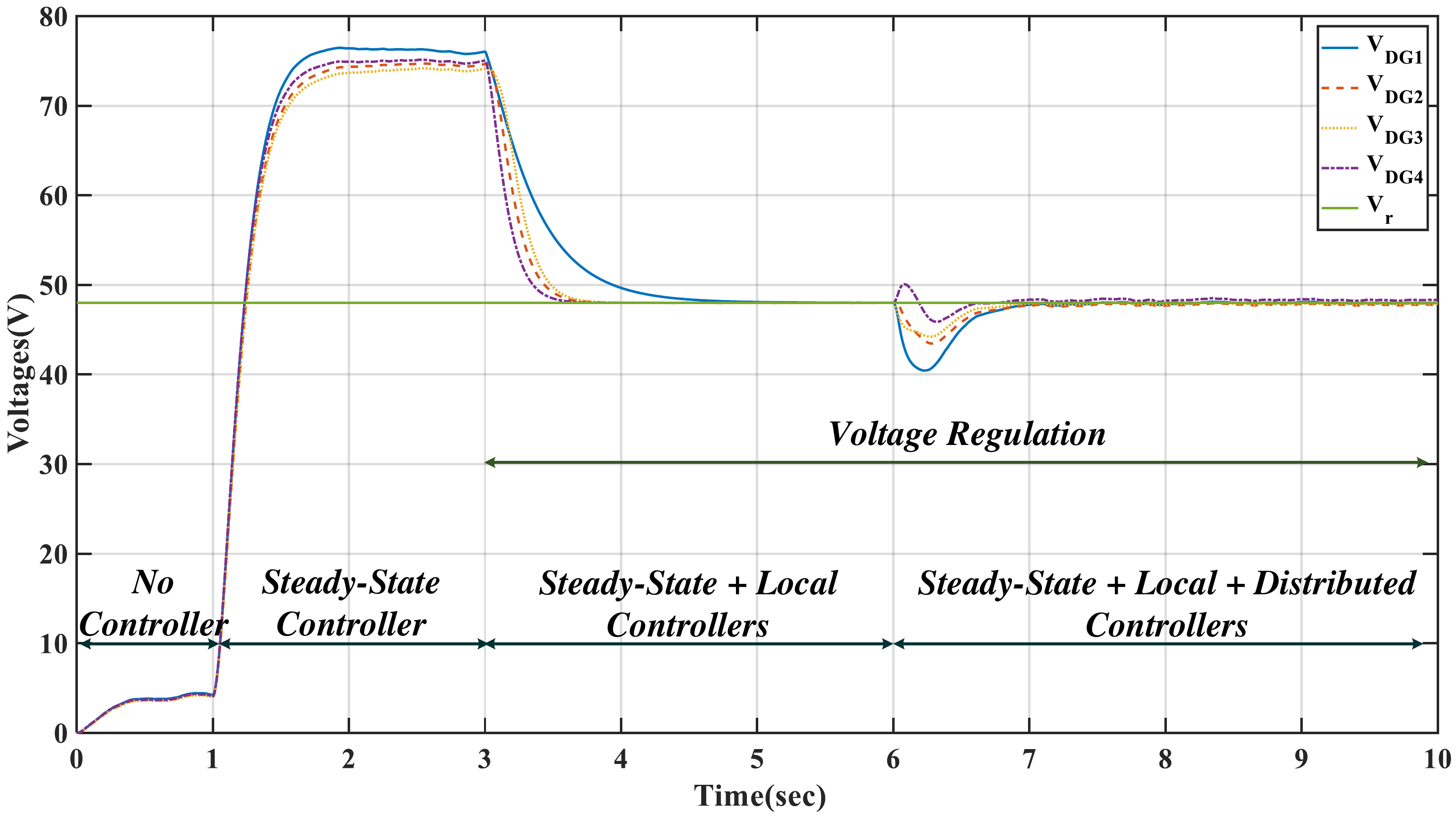}
    \caption{The output voltage magnitudes of DGs using proposed dissipativity-based control.}
    \label{fig.outputvoltage}
\end{figure}

\begin{figure}
    \centering
    \includegraphics[width=0.9\columnwidth]{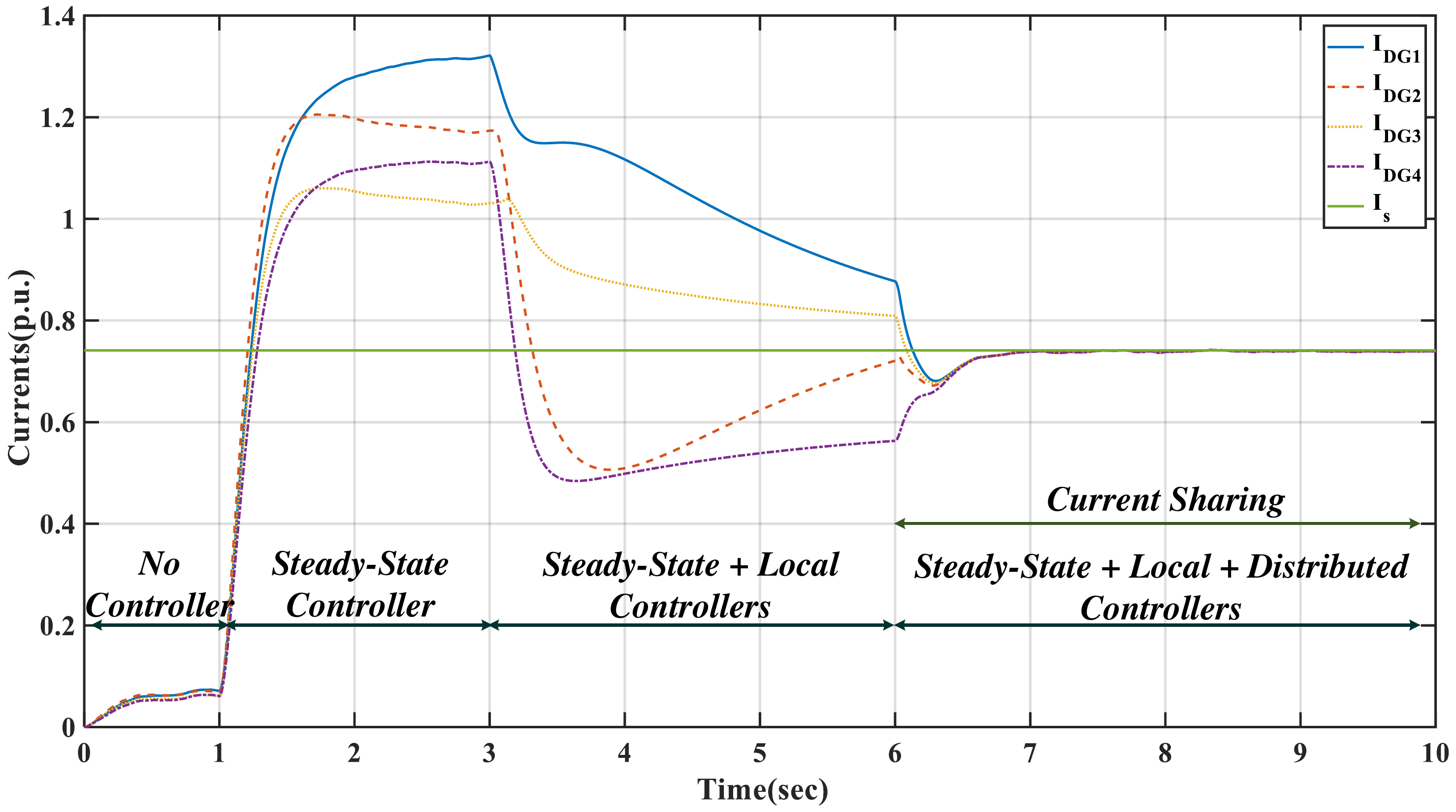}
    \caption{The output per-unit current magnitudes of DGs using proposed dissipativity-based control.}
    \label{fig.outputperunitcurrent}
\end{figure}

To ensure proportional current sharing among the DGs, the distributed control layer is activated at $t=6$ s. Figure \ref{fig.outputperunitcurrent} shows that after the introduction of distributed global control, the per-unit currents converge to an index value of $I_s=0.735$, thus ensuring proper current sharing among DGs. This convergence value closely matches the optimized value $I_s=0.740$ obtained from Th. \ref{Th:VRegulation_CSharing}, validating the proposed optimization framework.

These sequential activation results demonstrate the contribution of our hierarchical control strategy. The local control layer successfully achieves voltage regulation, bringing voltages to their reference values. Subsequently, the distributed control layer ensures proper current sharing without compromising voltage regulation performance. This hierarchical approach effectively coordinates local and distributed control objectives to meet the dual requirements of voltage regulation and current sharing in DC MGs.

\subsection{Load Changes}
To evaluate the robustness of the proposed controller, each DG was subjected to a series of load variations. Specifically, at $t=2$ s, a constant current load $\bar{I}_L$ was added to all DGs, followed by an increase in the existing constant impedance load $Y_L$ at $t=4$ s, which was subsequently restored to its original value at $t=8$ s. 

Fig. \ref{fig.loadchange}(a) shows that throughout these load variations, the DG output voltages consistently follow the reference value $V_r$ of 48 V with minimal transient deviations. The output voltages exhibit rapid restoration to the reference value after each disturbance, with settling times typically under 0.5 s. This fast voltage regulation performance underscores the effectiveness of the local control layer in maintaining voltage stability despite significant load perturbations.

Fig. \ref{fig.loadchange}(b) illustrates the current sharing performance of the proposed controller. Despite abrupt load changes, the per-unit currents of all DGs maintain proportional sharing according to their rated power. This balanced current sharing is particularly notable during the transient periods following load changes, where the controller quickly redistributes the load currents among DGs without compromising voltage regulation.

\begin{figure}
    \centering
    \includegraphics[width=0.9\columnwidth]{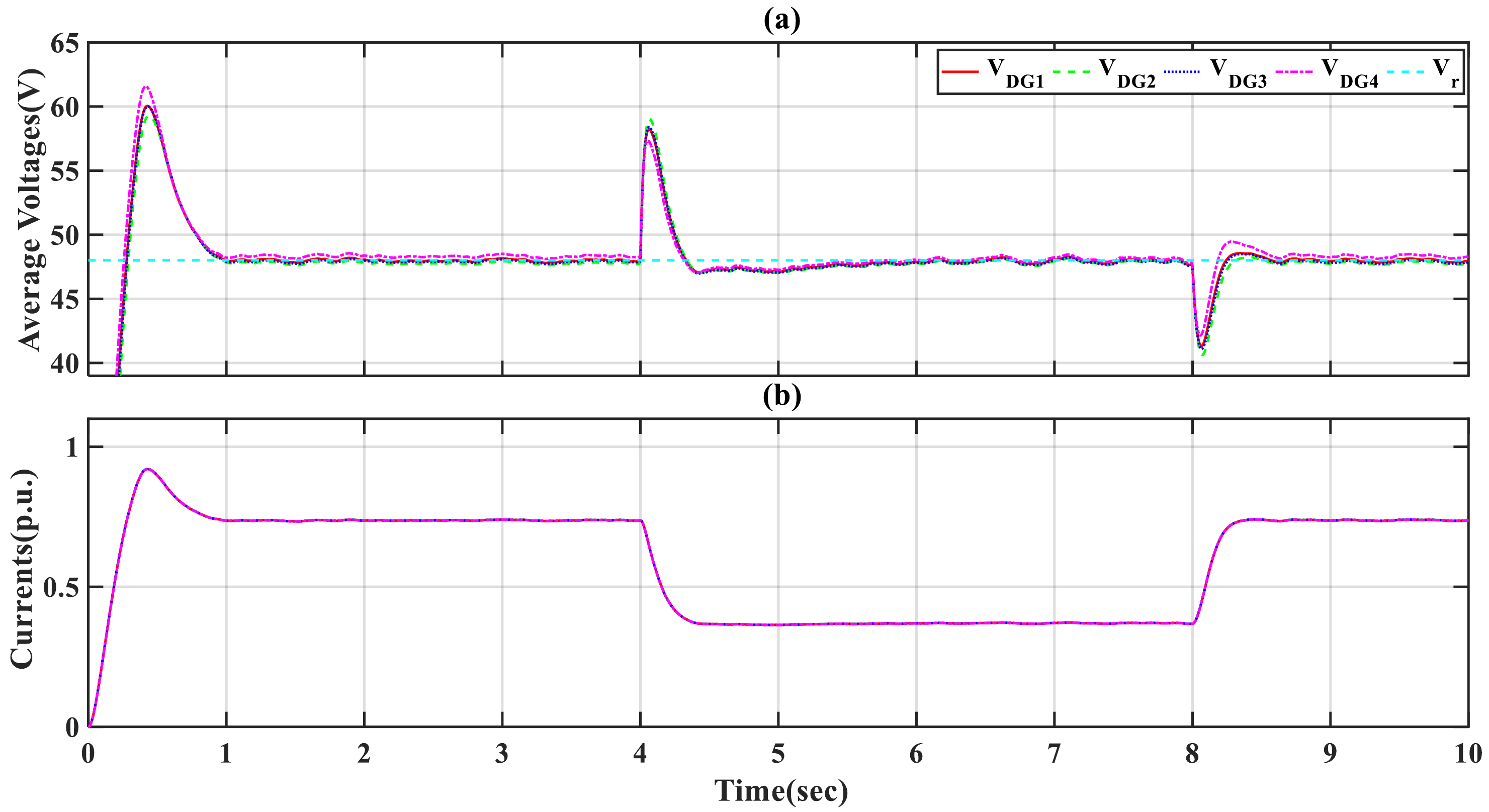}
    \caption{The DG output (a) voltages and (b) per-unit currents under load changes using dissipativity-based controllers in DC MG.}
    \label{fig.loadchange}
\end{figure}

\subsection{Communication Topology Co-design}

In this part, we considered two co-design variants with (i) a hard graph constraint that aligns the communication topology $\mathcal{G}^c$ with the physical topology $\mathcal{G}^p$, and (ii) a soft graph constraint that only penalizes the use of communication links that do not align with the physical topology $\mathcal{G}^p$. Therefore, the proposed co-design framework (Th. \ref{Th:CentralizedTopologyDesign}) was applied using two different approaches to shape the resulting communication topology. In the first scenario, we include an additional constraint in the co-design process to restrict the resulting communication topology to exactly match the physical topology (but with bi-directional links). Therefore, this scenario designs the communication topology under a ``hard graph constraint.'' In the second scenario, we omit such hard graph constraints but penalize the use of communication links. Therefore, this scenario designs the communication topology under ``soft graph constraints.'' 

\begin{figure}
    \centering
    \includegraphics[width=0.9\columnwidth]{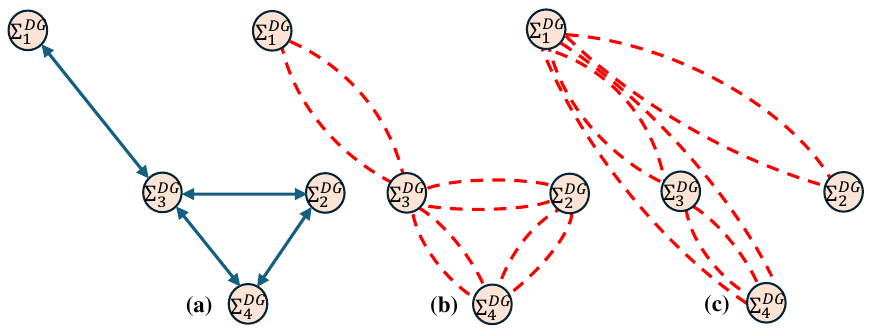}
    \caption{The (a) physical topology, (b) hard, and (c) soft communication graph constraints for DC MGs.}
    \label{fig.communicationtopology}
\end{figure}

Figure \ref{fig.communicationtopology} shows the physical topology and communication topologies under hard and soft graph constraints. The hard constraint approach in Fig. \ref{fig.communicationtopology}(b) limits communication links to physically adjacent DGs only. The soft constraint implementation in Fig. \ref{fig.communicationtopology}(c) creates optimized communication links between distant DGs (such as DG1-DG2 and DG1-DG4). This shows that the controller can identify the most beneficial communication connections for MG performance, even when the DGs are physically far apart and would normally have poor coordination.

The resulting controller gain matrices reveal substantial differences between the two approaches. As illustrated in Fig. \ref{fig.compare_Com}, the hard constraint approach generates significantly higher gain values, with maximum gains approximately 5 times higher than the soft constraint case. The quantitative comparison shows maximum gains of 1268 for hard constraints versus 270 for soft constraints, and average gains of 826 versus 117, respectively. Lower gain magnitudes in the soft constraint approach offer significant advantages by reducing control effort, reducing energy consumption, and decreasing sensitivity to measurement noise.

\begin{figure}
    \centering
    \includegraphics[width=0.9\columnwidth]{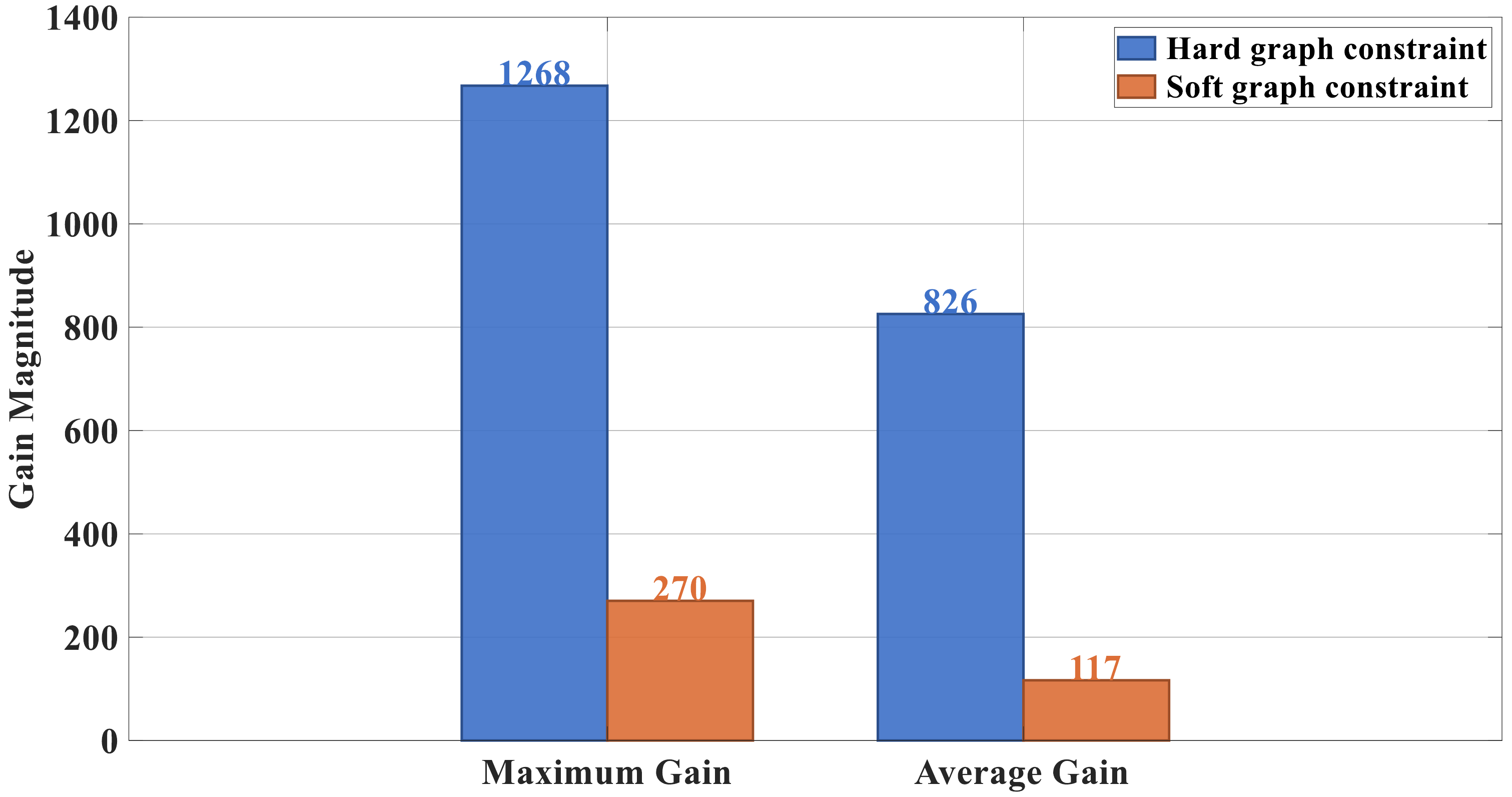}    
    \caption{Comparison of controller gain metrics between hard and soft constraint approaches.}
    \label{fig.compare_Com}
\end{figure}

Dynamic performance during load changes demonstrates that both approaches successfully regulate voltages to the reference value of 48 V. However, as shown in Fig. \ref{fig.ComTop_Perf}, the implementation of soft constraint exhibits superior disturbance rejection capabilities, achieving reduced voltage deviations and faster settling times. The soft constraint approach delivers better dynamic performance while utilizing substantially lower controller gains. This improvement highlights the advantages of an optimized communication topology that establishes strategic links between physically distant DGs. By connecting DGs that benefit the most from coordination rather than following the physical topology, the control system achieves better information exchange with less aggressive control actions.

\begin{figure}
    \centering
    \includegraphics[width=0.9\columnwidth]{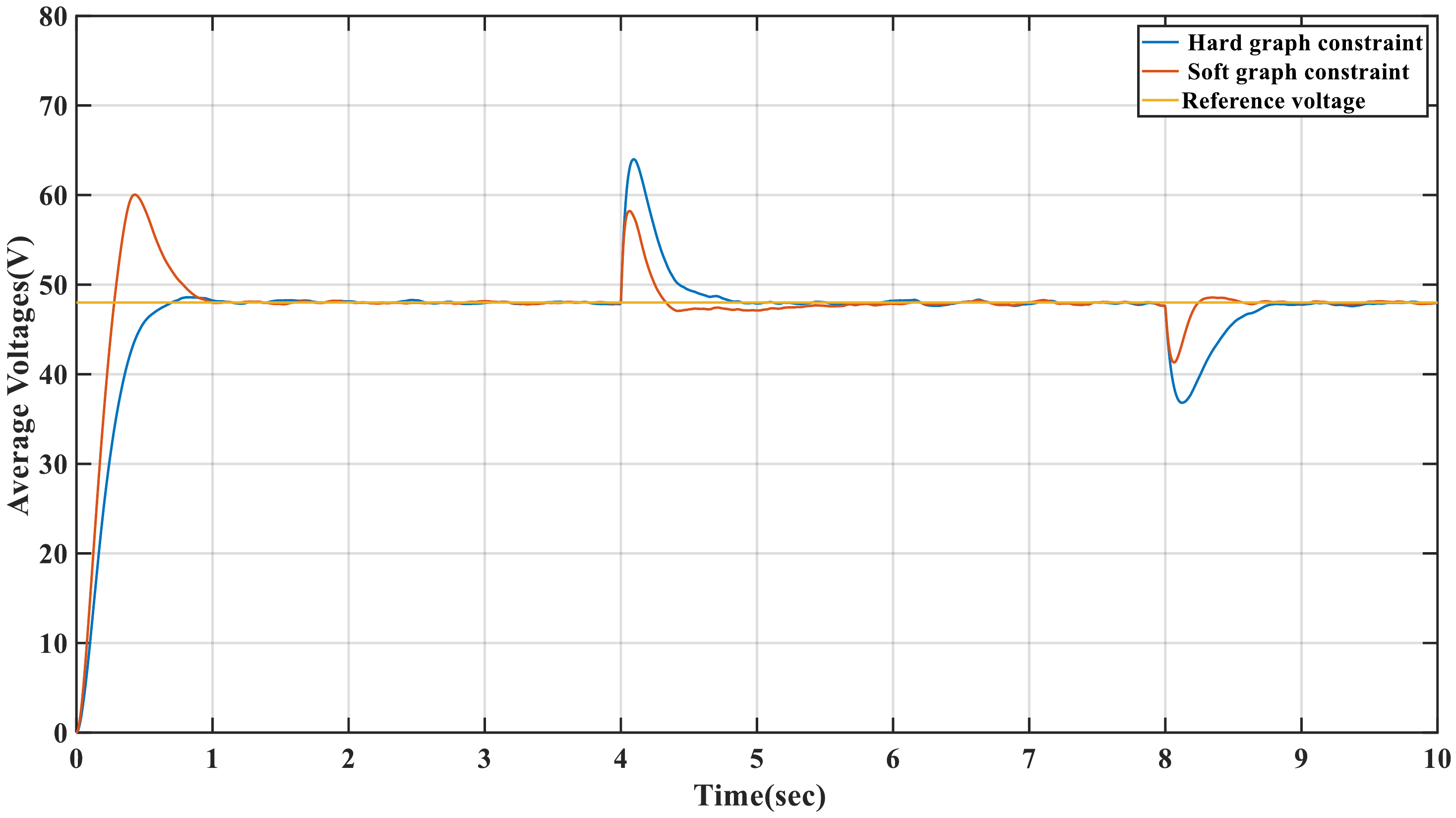}
    \caption{The DG output average voltages under load changes with hard and soft constraint controllers.}
    \label{fig.ComTop_Perf}
\end{figure}

\subsection{Disturbance Rejection}
To evaluate the robustness of our dissipativity-based control, we incorporated various disturbances into the DC MG. For each DG, we include voltage and current disturbances as described in \eqref{DGEQ}, modeled as zero-mean Gaussian noises with variances of $\sigma_{vi}^2 = 0.5$ and $\sigma_{ci}^2 = 0.5$, respectively.
Furthermore, we introduced line disturbances as discussed in \eqref{line}, where the uncertainty of line resistance $\Delta R_l(t)$ manifests itself as the disturbance term $\bar{w}_l(t)$ with variance $\sigma_l^2 = 0.5$. 

The simulation sequence was designed to evaluate the performance of the proposed controller with all disturbances active throughout the simulation period, including during load changes at $t=4$ s and $t=8$ s. Fig. \ref{fig.dist_Rej}(a) illustrates the comparison between the average voltages with and without disturbances, and Fig. \ref{fig.dist_Rej}(b) shows the corresponding comparison for current sharing performance. Despite the relatively high magnitude of disturbance, the results demonstrate that the proposed controller achieves robust performance with only minor steady-state oscillations of approximately $\pm1$ V.

\begin{figure}
    \centering
    \includegraphics[width=0.9\columnwidth]{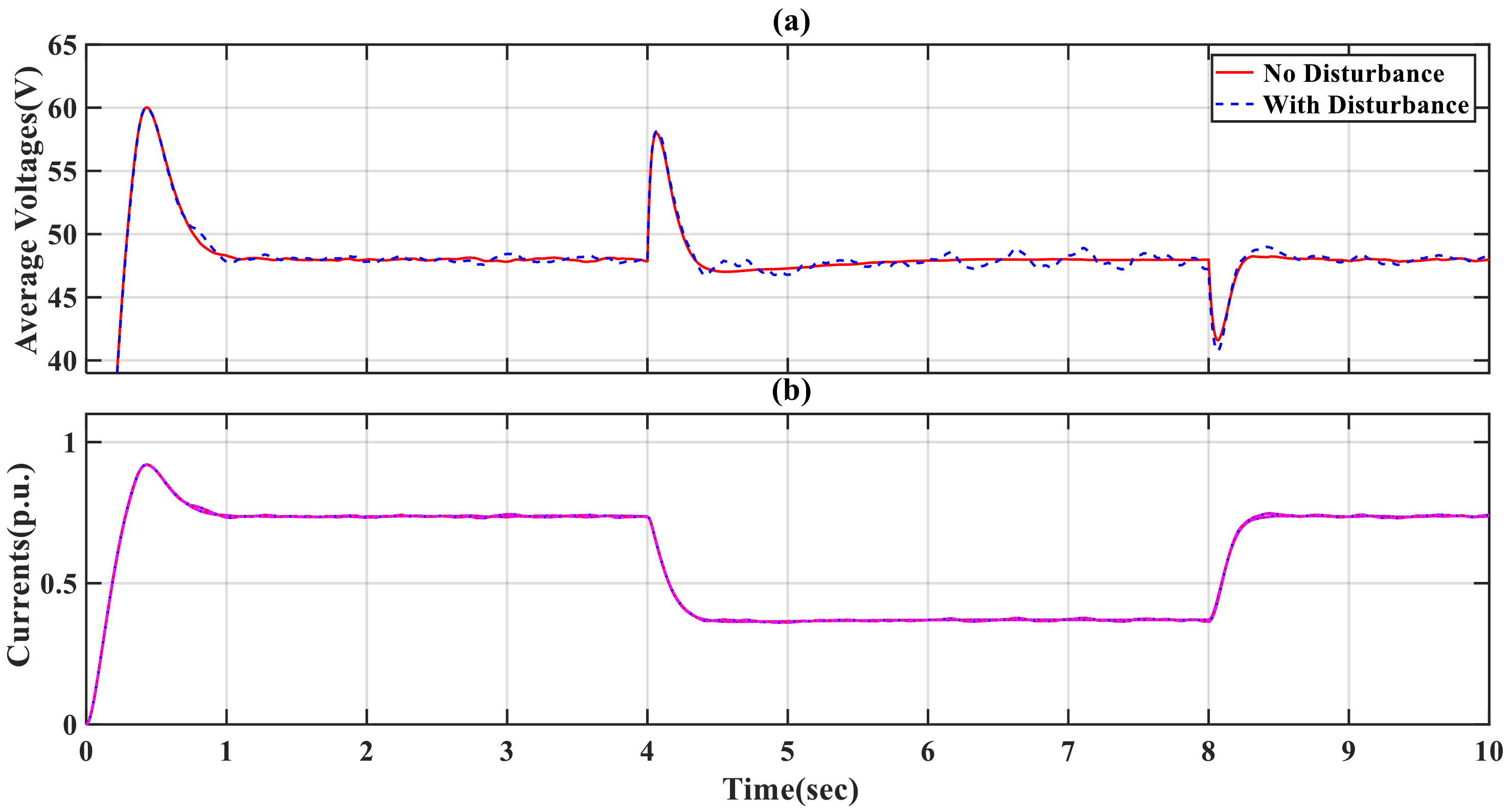}
    \caption{Disturbance rejection performance of the proposed controller. (a) Average DG voltages and (b) current sharing performance with and without disturbances.}
    \label{fig.dist_Rej}
\end{figure}

The dissipativity-based control framework successfully maintains both the voltage regulation and current sharing objectives even in the presence of significant disturbances. While some initial oscillations are observed during transient periods, the controller quickly dampens these effects and recovers to the reference values. This robust performance can be attributed to the inherent disturbance rejection capabilities of the dissipativity-based approach, which explicitly accounts for system uncertainties in the control design process.

\subsection{Comparison with Traditional Droop Control}
We conducted a comprehensive performance evaluation of the proposed dissipativity-based controller and conventional droop-based control for voltage regulation in the DC MG. To ensure a fair comparison, the dissipativity-based co-design process employed hard graph constraints, which resulted in a communication topology identical to the physical topology. All DGs and transmission line parameters remained consistent in both methods. The gains of the droop controller were tuned to optimize performance, using the droop controller from a well-established paper in \cite{guo2018distributed}.

Figure \ref{fig.droop_control} presents a comparison of the average voltage performance between the droop control and the proposed droop-free dissipativity-based control. Both the proposed dissipativity-based controller and the droop control demonstrate voltage overshoots during load changes, but with significant differences in magnitude. The droop control exhibits significantly larger voltage peaks compared to the more moderate overshoot of the dissipativity-based control. These substantial voltage variations demonstrate the limited ability of the droop control to respond effectively to dynamic load changes. Furthermore, droop control initially exhibits a voltage drop that requires the activation of a secondary control mechanism \cite{Najafirad2} at $t=1$ s to restore the voltage to the reference value.

This comparison underscores the limitations of the droop control approach, which requires additional control layers, i.e.,  secondary or distributed control, to compensate for both voltage deviations and overshoot. The proposed dissipativity-based controller maintains stable voltage regulation without additional control layers and provides smoother voltage restoration with better performance during load changes. These results demonstrate the superiority of the proposed controller over conventional droop controllers and establish it as a more effective solution for DC MG applications.

\begin{figure}
    \centering
    \includegraphics[width=0.9\columnwidth]{Figures/Droop_control.pdf}
    \caption{Comparison of average voltage regulation between the proposed dissipativity-based controller and droop controllers for DC MG.}
    \label{fig.droop_control}
\end{figure}

\subsection{Scalability}
In this section, we examine a system with significantly increased complexity compared to the 4-DG case, providing a realistic evaluation of performance in large-scale applications. To evaluate scalability, we applied the proposed controller to a larger DC MG with 10 DGs and 24 transmission lines. Figure \ref{fig.scalability} shows the electrical configuration of this expanded DC MG. This comprehensive test covers all previous scenarios, including load changes, disturbances, and communication topology co-design.

Figure \ref{fig.scalablity_communication} compares communication topologies under hard and soft graph constraints for the 10-DG DC MG. The hard constraint approach in Fig. \ref{fig.scalablity_communication}(a) restricts communication links to follow the physical network structure. However, the soft constraint method in Fig. \ref{fig.scalablity_communication}(b) creates optimized links between distant DGs. This optimization identifies the communication links that are the most beneficial for system performance and enhances overall system stability.

Figure \ref{fig.10DGs} confirms robust performance in large-scale DC MG. Fig. \ref{fig.10DGs}(a) shows that all out voltages follow the 48 V reference with fast restoration after disturbances. Fig. \ref{fig.10DGs}(b) demonstrates proper current sharing where the per-unit currents maintain a proportional distribution based on the power ratings. These results confirm that the controller successfully handles increased system complexity while achieving both voltage regulation and balanced current sharing objectives.

\begin{figure}
    \centering
    \includegraphics[width=0.9\columnwidth]{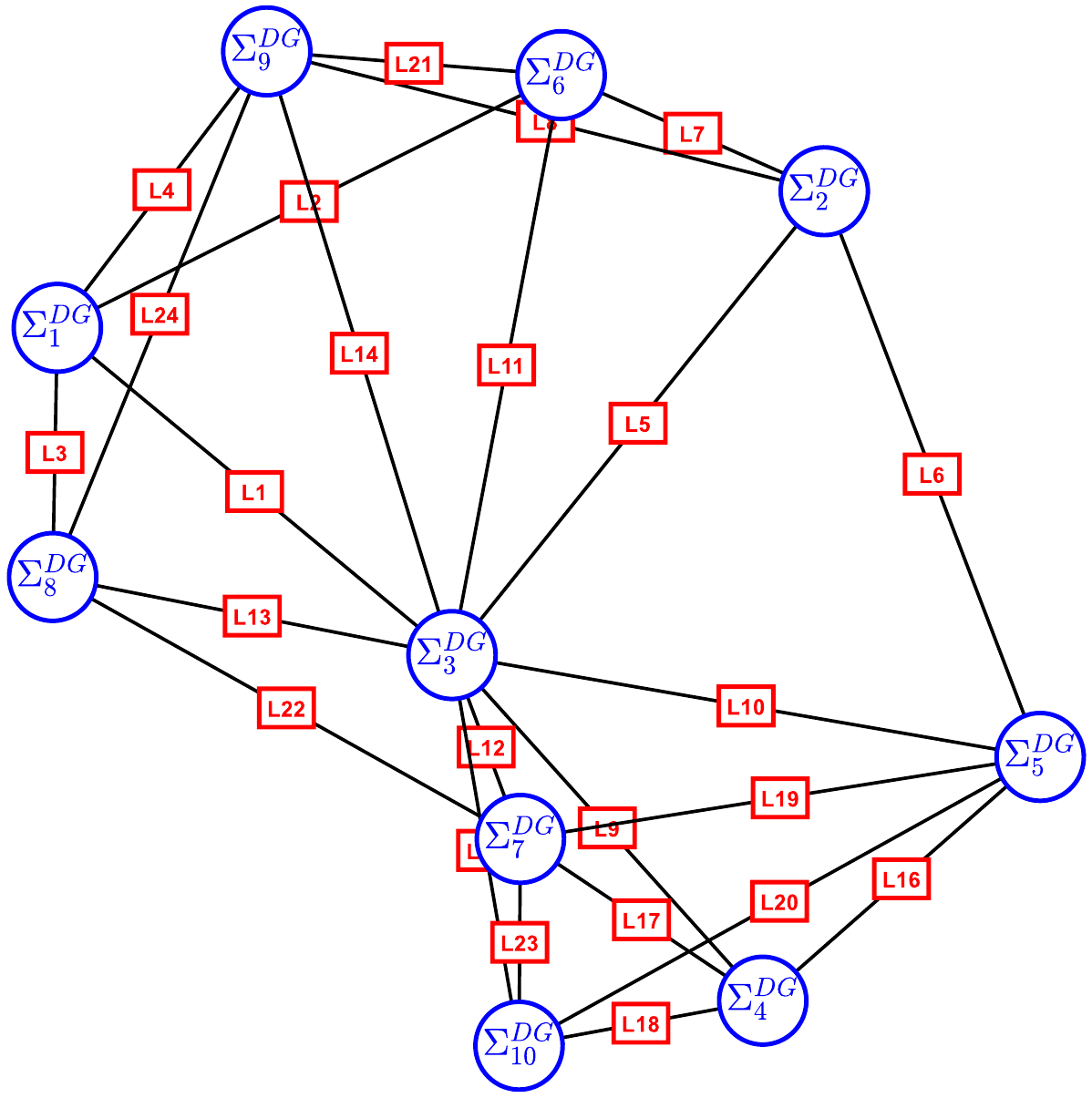}
    \caption{The physical topology of DC MG with 10 DGs and 24 transmission lines.}
    \label{fig.scalability}
\end{figure}

\begin{figure}
    \centering
    \includegraphics[width=1\columnwidth]{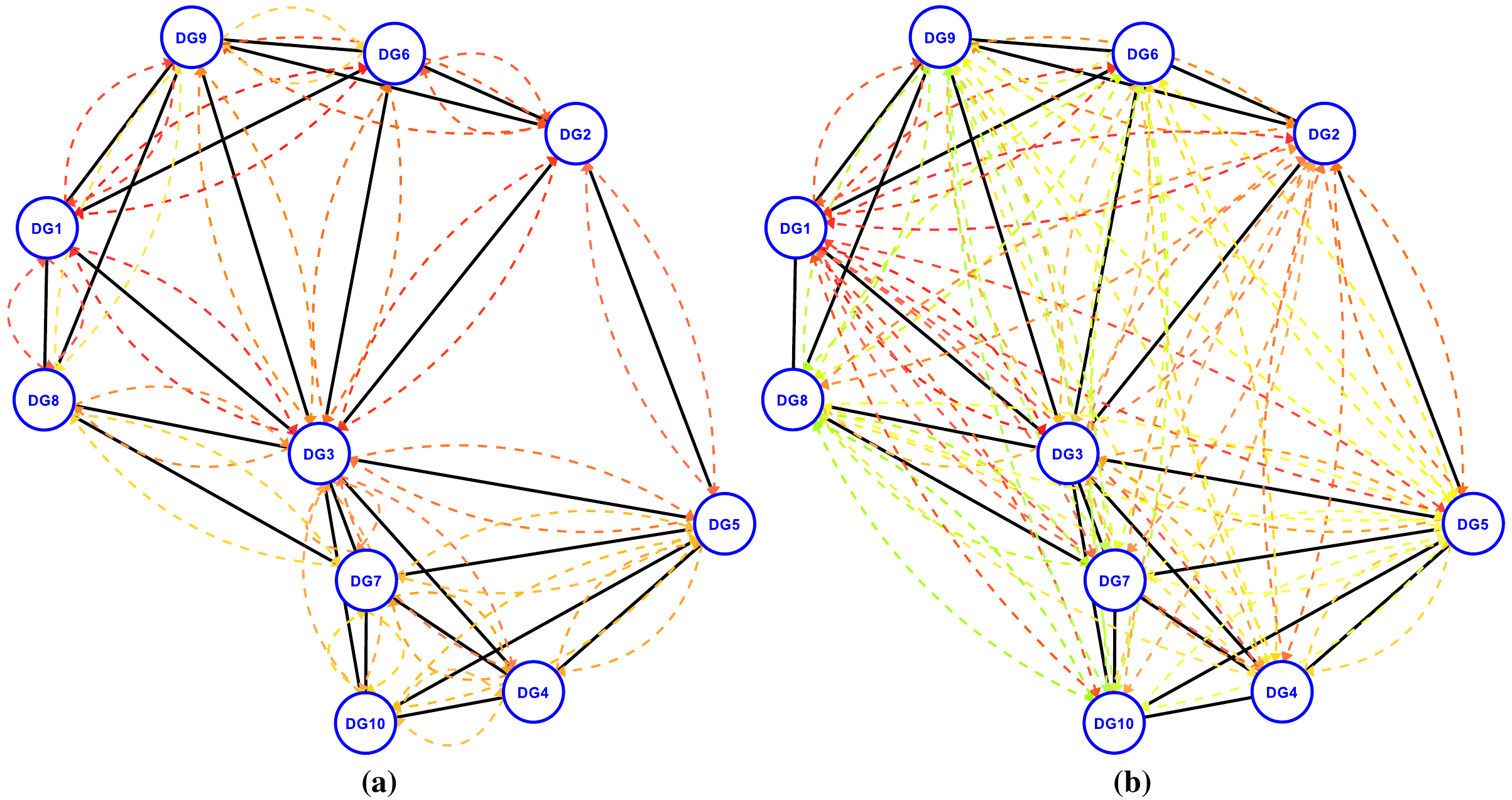}
    \caption{The communication topology for DC MG shown in Fig. \ref{fig.scalability} under (a) hard and (b) soft graph constraints.}
    \label{fig.scalablity_communication}
\end{figure}

\begin{figure}
    \centering
    \includegraphics[width=1\columnwidth]{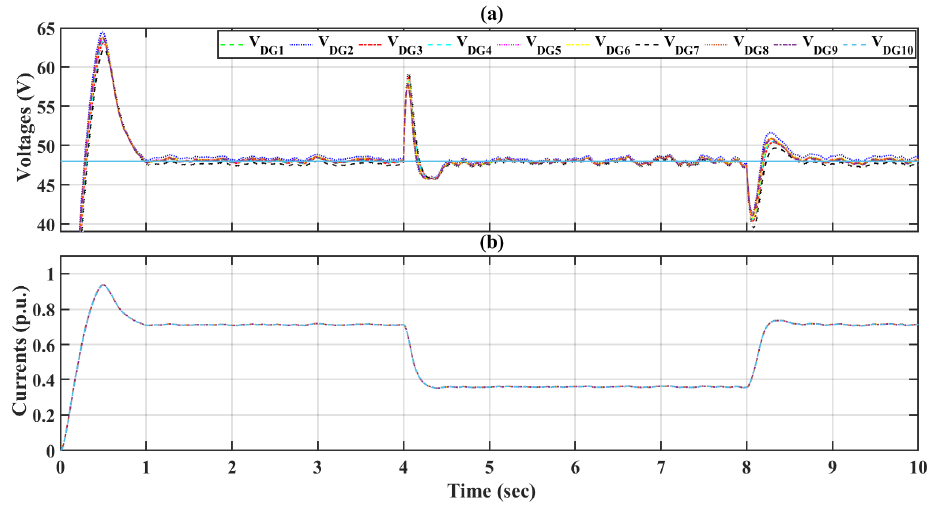}
    \caption{The DG output (a) voltages and (b) per-unit currents using the proposed dissipativity-based controllers in DC MG with 10 DGs and 24 transmission lines.}
    \label{fig.10DGs}
\end{figure}

\section{Conclusion}\label{Conclusion}
This paper proposes a dissipativity-based distributed control and communication topology co-design approach for DC MGs. By applying dissipativity theory, we develop a unified framework that simultaneously addresses distributed global controller and communication topology design problems in networked DC MG. To ensure the feasibility of this global co-design process, we employ carefully designed local controllers at each subsystem that guarantee local dissipativity properties and support the overall network coordination.
We formulate all design problems as LMI-based convex optimization problems to enable efficient and scalable evaluations. 
The co-design process establishes an optimized communication network that enables effective information exchange between DGs. The proposed droop-free controller operates on this optimized topology to ensure robust voltage regulation and current sharing under various disturbances. This approach shows significant advantages compared to conventional droop control methods. Future work will focus on developing a decentralized and compositional co-design framework that enables plug-and-play operation.

\bibliographystyle{IEEEtran}
\bibliography{References}

\end{document}